\documentclass[onecolumn,10pt]{article}
\usepackage[top=2.54cm, bottom=2.54cm, left=2.54cm, right=2.54cm,a4paper]{geometry}
\def\mytitle{Uncertainty relations: An operational approach to the error-disturbance tradeoff}

\usepackage[utf8]{inputenc} 
\usepackage{amsmath,amsthm}
\usepackage{bm}  % Define \bm{} to use bold math fonts
\usepackage{bbm}
\usepackage{verbatim}
\usepackage{xcolor}
\usepackage{mdframed}
\usepackage{units}
\usepackage{caption}
\DeclareMathOperator*{\minimum}{minimum}
\DeclareMathOperator*{\maximum}{maximum}

\usepackage[backend=biber,sorting=none,style=phys,biblabel=brackets,backref=true,bibencoding=utf8,maxnames=20,eprint=true]{biblatex}
\makeatletter
\pretocmd{\blx@head@bibintoc}{\phantomsection}{}{\ddt}
\makeatother
\DefineBibliographyStrings{english}{%
  backrefpage = {page},% originally "cited on page"
  backrefpages = {pages},% originally "cited on pages"
}
\usepackage{titlesec}

\titleformat*{\section}{\bfseries}
\titleformat*{\subsection}{\normalsize\bfseries}
\titleformat*{\subsubsection}{\bfseries}
\titleformat*{\paragraph}{\large\bfseries}
\titleformat*{\subparagraph}{\large\bfseries}

\titlespacing\section{0pt}{12pt plus 4pt minus 2pt}{2pt plus 2pt minus 2pt}

\usepackage{setspace} %load before hyperref...

\usepackage[bookmarks,colorlinks,breaklinks]{hyperref}  % PDF hyperlinks, with coloured links
\definecolor{dullmagenta}{rgb}{0.5,0,0.5}   % #660066
\definecolor{darkblue}{rgb}{0,0,0.6}
\definecolor{lessdarkblue}{rgb}{0,0,0.85}

\definecolor{darkred}{rgb}{.9,0,0}
\definecolor{quantumviolet}{HTML}{53257F} %Quantum violet
  \definecolor{quantumgray}{HTML}{555555} %Quantum gray
\hypersetup{linkcolor=darkred,citecolor=lessdarkblue,filecolor=dullmagenta,urlcolor=darkblue,
  pdftitle={\mytitle},
  pdfauthor={Joseph M.\ Renes, Volkher B.\ Scholz, and Stefan Huber}
  } % coloured links

\RequirePackage[charter, greekuppercase=italicized,cal=cmcal]{mathdesign}

\usepackage{amsbsy,verbatim,graphicx}

\newcommand{\gnorm}[2]{\left\|#1\right\|_{#2}}

\newcommand{\cbnorm}[1]{\gnorm{#1}{\rm cb}}

\newcommand{\opnorm}[1]{\gnorm{#1}{\infty}}

\newcommand*{\half}{\frac{1}{2}}

\usepackage{braket}
\newcommand{\ketbra}[1]{|#1\rangle\langle #1|}
\def\tr{{\rm Tr}}

\def\eps{\varepsilon}
\def\cA{\mathcal A}
\def\cB{\mathcal B}
\def\cC{\mathcal C}
\def\cD{D}
\def\cE{\mathcal E}
\def\cF{\mathcal F}
\def\cH{\mathcal H}
\def\cI{\mathcal I}
\def\cK{\mathcal K}
\def\cM{\mathcal M}
\def\cN{\mathcal N}
\def\cP{\mathcal P}
\def\cQ{\mathcal Q}
\def\cR{\mathcal R}

\def\cT{\mathcal T}
\def\cV{V}
\def\cW{\mathcal W}

\def\sC{\mathtt C}
\def\sX{\mathsf X}
\def\sY{\mathsf Y}
\def\sZ{\mathsf Z}
\def\sQ{\mathsf Q}

\renewcommand{\rho}{\varrho}
\renewcommand{\phi}{\varphi}
\def\id{\mathbbm 1}
\renewcommand{\circ}{}

\theoremstyle{plain}

 \newmdtheoremenv[%
  outerlinewidth=0,%
  linewidth = 0.0pt,%
  innertopmargin = 0
  ]{lemma}{Lemma}

   \newmdtheoremenv[%
  outerlinewidth=0,%
  linewidth = 0.0pt,%
  innertopmargin = 0
  ]{theorem}{Theorem}

 \newmdtheoremenv[%
  outerlinewidth=0,%
  linewidth = 0.0pt,%
  innertopmargin = 0
  ]{corollary}{Corollary}

\usepackage{tikz}
 \usetikzlibrary{external}
 \usetikzlibrary{arrows}
 \usetikzlibrary{decorations.pathmorphing} % for snake lines

\usepackage{pgfplots}
 \pgfplotsset{compat=newest}
  \usepgfplotslibrary{fillbetween}

\tikzstyle{block} = [draw,rectangle,thick,minimum height=2em,minimum width=2em]
\tikzstyle{qw} = [decorate, decoration={snake, amplitude=1mm,post length=0.7mm},thick,->]
\tikzstyle{cw} = [dashed,thick, ->]
\tikzstyle{pt} = [inner sep=0pt]

\usetikzlibrary{calc}
  \newlength{\@eQ}%Capital Q height
  \newlength{\@w}%line width
  \newlength{\@rl}%rounding length
  \newlength{\@cw}%character width
  \newlength{\@ch}%lower case character height
  \newlength{\@cr}%corner radius
  \newlength{\@sl}%<> slant
  \newlength{\@xt}%<> x thickeness
  
  \DeclareRobustCommand{\Quantum}{%
    {\sffamily%\color{quantumviolet}%
      \setlength{\@eQ}{\dimexpr\fontcharht\font`Q\relax}%
      \setlength{\@w}{0.088\@eQ}%
      \setlength{\@rl}{0.2\@eQ}%
      \setlength{\@cw}{0.5\@eQ}%
      \setlength{\@ch}{0.65\@eQ}%
      \setlength{\@cr}{0.3\@w}%
      \setlength{\@sl}{0.22\@eQ}%
      \setlength{\@xt}{1.113588507968435\@w}%=math.sqrt(1/(1-pow(22./50.,2)))*\@w
      \tikz[baseline,x=\@eQ,y=\@eQ,every node/.append style={fill=none,inner sep=0pt,outer sep=0pt,node distance=0},rounded corners=\@cr]{%
        \node[overlay,anchor=base west,opacity=0] {Q};
        \fill[] (0,0.5) -- ++(\@sl,0.5) -- ++(\@xt,0) -- ++(-\@sl,-0.5) -- ++(\@sl,-0.5) -- ++(-\@xt,0) -- cycle;
        \begin{scope}[xshift=\@eQ]
          \fill (0,0) -- ++(-\@sl,0.5) -- ++(-\@xt,0) -- ++(\@sl,-0.5) -- cycle;    
          \clip[overlay,rounded corners=0] (-1.34\@xt,0) -- ++(-\@sl,0.5) -- ++(0,-0.5) -- cycle (0.34\@xt,0) -- ++($2*(-\@sl,0.5)$) -- (0.5\@xt,1) --cycle;
          \fill (0,0.5) node (eastend) {} -- ++(-\@sl,0.5) -- ++(-\@xt,0) -- ++(\@sl,-0.5) -- ++(-\@sl,-0.5) -- ++(\@xt,0) -- cycle;
        \end{scope}
        \def\@u##1{%
          \fill[##1] (0,\@ch) -- (0,1.06\@rl) to[out=-90,in=184,looseness=1.2,overlay] ($(\@cw-\@w,0)$) -- ($(\@cw,0)$) -- ++(0,\@ch) -- ++(-\@w,0) -- ++($(0,-\@ch) + (0,\@w)$) -- ($(\@cw-\@w,\@w)$) to[out=184,in=-90,looseness=1.1] ($(\@w,1.06\@rl)$) -- (\@w,\@ch) --cycle;}
        \@u{shift={($(current bounding box.south east)+(0.172\@eQ,0)$)}}
        \node[overlay,anchor=base east,opacity=0] at (current bounding box.south east) {u};
        \fill[shift={($(current bounding box.south east)+(0.110\@eQ,0)$)}] (\@cw,0) -- ++($(0,\@ch)-(0,\@rl)$) to[out=90,in=0,looseness=1.22] ++($(-1.22\@rl,\@rl)$) -- ($(\@w,\@ch)$) -- ++(0,-\@w) -- ($(1.22\@rl,\@ch-\@w)$) to[out=0,in=90,looseness=1.3,rounded corners=0] ($(\@cw,\@ch)-(\@w,\@rl)$) -- ($(\@cw,0)-(\@w,-\@w)$) -- 
        ($(\@rl,\@w)$) to[out=180,in=-100,looseness=0.95,rounded corners=0] ($(1.05\@w,0.33\@ch)$) to[out=80,in=176,looseness=1.0] ($(\@cw-\@w,0.47\@ch)$) to[rounded corners=0] ++(0,-\@cr) to[rounded corners=0] ++($(0,\@cr+\@w)$) to[out=176,in=0] ($(1.2\@rl,0.48\@ch)+(0,\@w)$) to[out=180,in=90,rounded corners=0] ($(0.05\@w,0.33\@ch)$) to[out=-90,in=180,looseness=1.1,rounded corners=0] ($(\@rl,0)$) -- cycle;
        \node[overlay,anchor=base east,opacity=0] at (current bounding box.south east) {a};
        \fill[shift={($(current bounding box.south east)+(0.181\@eQ+\@cw,\@ch)$)},rotate=180] (0,\@ch) -- (0,\@rl) to[out=-90,in=180,looseness=1.22] ($(1.2\@rl,0)$) -- (\@cw,0) -- ++(0,\@ch) -- ++(-\@w,0) -- ++($(0,-\@ch) + (0,1.05\@w)$) to[out=182,in=-90,looseness=1.04] ($(\@w,\@rl)$) -- (\@w,\@ch) --cycle;
        \node[overlay,anchor=base east,opacity=0] at (current bounding box.south east) {n};
        \fill[shift={($(current bounding box.south east)+(0.111\@eQ,0)$)}] 
        (\@w,0.9) [rounded corners=0] -- (\@w,\@ch) [rounded corners=\@cr] -- (0,\@ch) -- ++(0,-\@w) to[rounded corners=0] ++(\@w,0) -- (\@w,\@rl) to[out=-90,in=180,looseness=1.3,overlay] ($(\@w+\@rl,-0.07\@w)$) to[overlay] ++($(0.64\@cw-\@w-\@rl,0)$) -- ++(0,\@w) -- ++($(\@w+\@rl-0.64\@cw,0)$) to[out=180,in=-90,looseness=1.4] ($(2\@w,\@rl)$) -- ++($(0,\@ch-\@rl-\@w)$) -- ($(0.64\@cw,\@ch-\@w)$) -- ++(0,\@w) to[rounded corners=0] (2\@w,\@ch) -- ($(2\@w,0.9)$) --cycle;
        \node[overlay,anchor=base east,opacity=0] at (current bounding box.south east) {t};
        \@u{shift={($(current bounding box.south east)+(0.125\@eQ,0)$)}}
        \node[overlay,anchor=base east,opacity=0] at (current bounding box.south east) {u};
        \fill[shift={($(current bounding box.south east)+(0.16\@eQ+1.93*\@cw-\@w,\@ch)$)},rotate=180] (0,\@ch) -- (0,\@rl) to[out=-90,in=180,looseness=1.2] ($(1.2\@rl,0)$) -- ($(1.93\@cw,0)-(\@w,0)$) -- ++(0,\@ch) -- ++(-\@w,0) -- ++($(0,-\@ch) + (0,1.05\@w)$) to[out=184,in=-4] ($(0.97\@cw,1.05\@w)+(0,0)$) -- ++($(0,-1.05\@w)+(0,\@ch)$) -- ++($(-\@w,0)$) -- ++($(0,-\@ch)+(0,1.05\@w)$) to[out=182,in=-90,looseness=1.04] ($(\@w,\@rl)$) -- (\@w,\@ch) --cycle;
        \node[overlay,anchor=base east,opacity=0] at (current bounding box.south east) {m};
      }}}

\usepackage{xfrac}

\usepackage{titling}

\posttitle{\par\end{center}}
\begin{document}

\title{\belowpdfbookmark{\mytitle}{}{\large {\bf \mytitle}}}

\author{
{\normalsize 
\href{http://orcid.org/0000-0003-2302-8025}{\color{black} Joseph M.\ Renes}\textsuperscript{1}, \href{http://orcid.org/0000-0002-3235-022X}{\color{black}Volkher B.\ Scholz}\textsuperscript{1,2}, and \href{http://orcid.org/0000-0002-0564-5436}{\color{black}Stefan Huber}\textsuperscript{1,3}}\\
\emph{\small 
\textsuperscript{1}Institute for Theoretical Physics, ETH Z\"urich, Switzerland}\\[-1mm]
\emph{\small \textsuperscript{2}Department of Physics, Ghent University, Belgium}\\[-1mm]
\emph{\small \textsuperscript{3}Department of Mathematics, Technische Universit\"at M\"unchen, Germany}
}

\def\acceptdate{{July 10, 2017}} 
\StrSubstitute{\mytitle}{ }{+}[\titleurl]
\date{\vspace{-2mm}\normalsize{{\color{quantumviolet}{\href{http://quantum-journal.org/?s=\titleurl}{\color{quantumviolet} Accepted in {\large \Quantum}\, \acceptdate }}}}\vspace{-\baselineskip}\vspace{-2mm}}
%\date{\vspace{-\baselineskip}}

\maketitle

%\myabstract{
\begin{abstract}
The notions of error and disturbance appearing in quantum uncertainty relations are often quantified by the discrepancy of a physical quantity from its ideal value. 
However, these real and ideal values are not the outcomes of simultaneous measurements, and comparing the values of unmeasured observables is not necessarily meaningful according to quantum theory. 
To overcome these conceptual difficulties, we take a different approach and define error and disturbance in an operational manner. 
In particular, we formulate both in terms of the probability that one can successfully distinguish the actual measurement device from the relevant hypothetical ideal by any experimental test whatsoever. 
This definition itself does not rely on the formalism of quantum theory, avoiding many of the conceptual difficulties of usual definitions.  
We then derive new Heisenberg-type uncertainty relations for both joint measurability and the error-disturbance tradeoff for arbitrary observables of finite-dimensional systems, as well as for the case of position and momentum. 
Our relations may be directly applied in information processing settings, for example to infer that devices which can faithfully transmit information regarding one observable do not leak any information about conjugate observables to the environment. 
We also show that Englert's wave-particle duality relation [\href{http://dx.doi.org/10.1103/PhysRevLett.77.2154}{Phys.\ Rev.\ Lett.\ {\bf 77}, 2154 (1996)}] can be viewed as an error-disturbance uncertainty relation.
\end{abstract}
\vspace{0.5\baselineskip}
%}

\section{Introduction}
It is no overstatement to say that the uncertainty principle is a cornerstone of our understanding of quantum mechanics, clearly marking the departure of quantum physics from the world of classical physics. 
Heisenberg's original formulation in 1927 mentions two facets to the principle. 
The first restricts the joint measurability of observables, stating that noncommuting observables such as position and momentum can only be simultaneously determined with a characteristic amount of indeterminacy~\cite[p.\ 172]{heisenberg_uber_1927} (see \cite[p.\ 62]{wheeler_quantum_1984} for an English translation).
%\footnote{``...es wird gezeigt, da\ss{} kanonisch konjugierte Gr\"o\ss{}en simultan nur mit einer charakteristischen Ungenauigkeit bestimmt werden k\"onnen.''\cite[p.\ 172]{heisenberg_uber_1927}} 
The second describes an error-disturbance tradeoff, noting that the more precise a measurement of one observable is made, the greater the disturbance to noncommuting observables~\cite[p.\ 175]{heisenberg_uber_1927} (\cite[p.\ 64]{wheeler_quantum_1984}). %\footnote{``Im Augenblick der Ortsbestimmung...ver\"andert das Elektron seinen Impuls unstetig. Diese \"Anderung ist um so gr\"o\ss{}er...je genauer die Ortsbestimmung ist.''\cite[p.\ 175]{heisenberg_uber_1927}}
The two are of course closely related, and Heisenberg argues for the former on the basis of the latter. 
Neither version can be taken merely as a limitation on measurement of otherwise well-defined values of position and momentum, but rather as questioning the sense in which values of two noncommuting observables can even be said to simultaneously exist. 
Unlike classical mechanics, in the framework of quantum mechanics we cannot necessarily regard unmeasured quantities as physically meaningful. 

More formal statements were constructed only much later, due to the lack of a precise mathematical description of the measurement process in quantum mechanics. 
Here we must be careful to draw a distinction between statements addressing Heisenberg's original notions of uncertainty from those, like the standard Kennard-Robertson uncertainty relation~\cite{kennard_zur_1927,robertson_uncertainty_1929}, which address the impossibility of finding a quantum state with well-defined values for noncommuting observables. 
Entropic uncertainty relations~\cite{maassen_generalized_1988,berta_uncertainty_2010} are also an example of this class; see \cite{coles_entropic_2015} for a review. 
Joint measurability has a longer history, going back at least to the seminal work of Arthurs and Kelly~\cite{arthurs_simultaneous_1965} and continuing in~\cite{she_simultaneous_1966,davies_quantum_1976,ali_systems_1977,prugovecki_fuzzy_1977,busch_indeterminacy_1985,busch_unsharp_1986,arthurs_quantum_1988,martens_towards_1991,ishikawa_uncertainty_1991,raymer_uncertainty_1994,leonhardt_uncertainty_1995,appleby_concept_1998,hall_prior_2004,werner_uncertainty_2004,ozawa_uncertainty_2004,watanabe_uncertainty_2011,busch_proof_2013,busch_heisenberg_2014,busch_measurement_2014}. %
Quantitative error-disturbance relations have only been formulated relatively recently, going back at least to Braginsky and Khalili~\cite[Chap.\ 5]{braginsky_quantum_1992} and continuing in~\cite{martens_disturbance_1992,appleby_concept_1998,ozawa_universally_2003,watanabe_quantum_2011,branciard_error-tradeoff_2013,buscemi_noise_2014,ipsen_error-disturbance_2013,coles_state-dependent_2015}.

Beyond technical difficulties in formulating uncertainty relations, there is a perhaps more difficult conceptual hurdle in that the intended consequences of the uncertainty principle seem to preclude their own straightforward formalization. 
To find a relation between, say, the error of a position measurement and its disturbance to momentum in a given experimental setup like the gamma ray microscope  would seem to require comparing the actual values of position and momentum with their supposed ideal values. 
However, according to the uncertainty principle itself, we should be wary of simultaneously ascribing well-defined values to the actual and ideal position and momentum since they do not correspond to commuting observables. 
Thus, it is not immediately clear how to formulate either meaningful measures of error and disturbance, for instance as mean-square deviations between real and ideal values, or a meaningful relation between them.\footnote{Uncertainty relations like the Kennard-Robertson bound or entropic relations do not face this issue as they do not attempt to compare actual and ideal values of the observables.} This question is the subject of much ongoing debate~\cite{ozawa_universally_2003,ozawa_uncertainty_2004-1,busch_proof_2013,busch_quantum_2014,appleby_quantum_2016,ozawa_disproving_2013}.

Without drawing any conclusions as to the ultimate success or failure of this program, in this paper we propose a completely different approach which we hope sheds new light on these conceptual difficulties.
Here, we define error and disturbance in an operational manner and ask for uncertainty 
relations that are statements about the properties of measurement devices, not of fixed experimental setups or of physical quantities themselves. 
More specifically, we define error and disturbance in terms of the \emph{distinguishing probability}, the probability that the actual behavior of the measurement apparatus can be distinguished from the relevant ideal behavior in any single experiment whatsoever.
To characterize measurement error, for example, we imagine a black box containing either the actual device or the ideal device. 
By controlling the input and observing the output we can make an informed guess as to which is the case. 
We then attribute a large measurement error to the measurement apparatus if it is easy to tell the difference, so that there is a high probability of correctly guessing, and a low error if not; of course we pick the optimal input states and output measurements for this purpose.
In this way we do not need to attribute a particular ideal value of the observable to be measured, we do not need to compare actual and ideal values themselves (nor do we necessarily even care what the possible values are), and instead we focus squarely on the properties of the device itself.
Intuitively, we might expect that calibration provides the strictest test, i.e.\ inputting states with a known value of the observable in question. 
But in fact this is not the case, as entanglement at the input can increase the distinguishability of two measurements. 
The merit of this approach is that the notion of distinguishability itself does not rely on any concepts or formalism of quantum theory, which helps avoid conceptual difficulties in formalizing the uncertainty principle.

Defining the disturbance an apparatus causes to an observable is more delicate, as an observable itself does not have a directly operational meaning (as opposed to the measurement of an observable).
But we can consider the disturbance made either to an ideal measurement of the observable or to ideal preparation of states with well-defined values of the observable. 
In all cases, the error and disturbance measures we consider are directly linked to a well-studied norm on quantum channels known as the completely bounded norm or diamond norm. 
We can then ask for bounds on the error and disturbance quantities for two given observables that every measurement apparatus must satisfy. 
In particular, we are interested in bounds depending only on the chosen observables and not the particular device. 
Any such relation is a statement about measurement devices themselves and is 
not specific to the particular experimental setup in which they are used. 
Nor are such relations statements about the values or behavior of physical quantities themselves. 
In this sense, we seek statements of the uncertainty principle akin to Kelvin's form of the second law of thermodynamics as a constraint on thermal machines, and not like Clausius's or Planck's form involving the behavior of physical quantities (heat and entropy, respectively). 
By appealing to a fundamental constraint on quantum dynamics, the continuity (in the completely bounded norm) of the Stinespring dilation~\cite{kretschmann_information-disturbance_2008,kretschmann_continuity_2008}, we find error-disturbance uncertainty relations for arbitrary observables in finite dimensions, as well as for position and momentum. 
Furthermore, we show how the relation for measurement error and measurement disturbance can be transformed into a joint-measurability uncertainty relation. 
Interestingly, we also find that Englert's wave-particle duality relation~\cite{englert_fringe_1996} can be viewed as an error-disturbance relation.

The case of position and momentum illustrates the stark difference between the kind of uncertainty statements we can make in our approach with one based on the notion of comparing real and ideal values. 
Take the notion of joint measurability, where we would like to formalize the notion that no device can accurately measure both position and momentum. 
In the latter approach one would first try to quantify the amount of position or momentum error made by a device as the discrepancy to the true value, and then show that they cannot both be small. 
The errors would be in units of position or momentum, respectively, and the hoped-for uncertainty relation would pertain to these values. 
Here, in contrast, we focus on the performance of the actual device relative to fixed ideal devices, in this case idealized separate measurements of position or momentum. 
Importantly, we need not think of the ideal measurement as having infinite precision. 
Instead, we can pick any desired precision and ask if the behavior of the actual device is essentially the same as this precision-limited ideal. 
Now the position and momentum errors do not have units of these quantities (they are unitless and always lie between zero and one), but instead \emph{depend on the desired precision}. 
Our uncertainty relation then implies that both errors cannot be small if we demand high precision in both position and momentum.  
In particular, when the product of the scales of the two precisions is small compared to Planck's constant, then the errors will be bounded away from zero (see Theorem~\ref{thm:infinitemerit} for a precise statement). 
It is certainly easier to have a small error in this sense when the demanded precision is low, and this accords nicely with the fact that sufficiently-inaccurate joint measurement is possible. 
Indeed, we find no bound on the errors for low precision. 

An advantage and indeed a separate motivation of an operational approach is that bounds involving operational quantities are often useful in analyzing information processing protocols. 
For example, entropic uncertainty relations, which like the Robertson relation characterize quantum states, have proven very useful in establishing simple proofs of the security of quantum key distribution~\cite{renes_conjectured_2009,berta_uncertainty_2010,tomamichel_uncertainty_2011,tomamichel_tight_2012,coles_entropic_2015}.
Here we show that the error-disturbance relation implies that quantum channels which can faithfully transmit information regarding one observable do not leak any information whatsoever about conjugate observables to the environment. 
This statement cannot be derived from entropic relations, as it holds for all channel inputs. 
It can be used to construct leakage-resilient classical computers from fault-tolerant quantum computers~\cite{lacerda_classical_2014}, for instance.

The remainder of the paper is structured as follows. 
In the next section we give the mathematical background necessary to state our results, and describe how the general notion of distinguishability is related to the completely bounded norm (cb norm) in this setting. 
In Section \ref{sec:definitions} we define our error and disturbance measures precisely. 
Section \ref{sec:finitedimresults} presents the error-disturbance tradeoff relations for finite dimensions, and details how joint measurability relations can be obtained from them. 
Section \ref{sec:positionmomentumresults} considers the error-disturbance tradeoff relations for position and momentum. 
Two applications of the tradeoffs are given in Section \ref{sec:app}: a formal statement of the information disturbance tradeoff for information about noncommuting observables and the connection between error-disturbance tradeoffs and Englert's wave-particle duality relations. 
In Section~\ref{sec:comparison} we compare our results to previous approaches in more detail, and finally we finish with open questions in Section~\ref{sec:openquestions}.

\section{Mathematical setup}
\subsection{Distinguishability}
The notion of the distinguishing probability is independent of the mathematical framework needed to describe quantum systems, so we give it first. 
Consider an apparatus $\cE$ which in some way transforms an input $A$ into an output $B$. 
To describe how different $\cE$ is from another such apparatus $\cE'$, we can imagine the following scenario. 
Suppose that we randomly place either $\cE$ or $\cE'$ into a black box such that we no longer have any access to the inner workings of the device, only its inputs and outputs.
Now our task is to guess which device is actually in the box by performing a single experiment, feeding in any desired input and observing the output in any manner of our choosing. 
In particular, the inputs and measurements can and should depend on $\cE$ and $\cE'$. 
The probability of making a correct guess, call it $p_{\rm dist}(\cE,\cE')$, ranges from $\frac 12$ to $1$, since we can always just make a random guess without doing any experiment on the box at all. 
Therefore it is more convenient to work with the distinguishability measure 
\begin{align}
\label{eq:distinguishdef}
\delta(\cE,\cE'):=2p_{\rm dist}(\cE,\cE')-1\,,
\end{align}
which ranges from zero (completely indistinguishable) to one (completely distinguishable). 
Later on we will show this quantity takes a specific mathematical form in quantum mechanics. 
But note that the definition implies that the distinguishability is monotonic under concatenation with a channel $\cF$ to both $\cE$ and $\cE'$, since this just restricts the possible tests. 
That is, both $\delta(\cE\circ\cF,\cE'\circ\cF)\leq \delta(\cE,\cE')$ and $\delta(\cF\circ\cE,\cF\circ \cE')\leq \delta(\cE,\cE')$ hold for all channels $\cF$ whose inputs and outputs are such that the channel concatenation is sensible. 
Here and in the remainder of the paper, we denote concatenation of channels by juxtaposition, while juxtaposition of operators denotes multiplication as usual.

\subsection{Systems, algebras, channels, and measurements}
\label{sec:systems}

In the finite-dimensional case we will be interested in two arbitrary nondegenerate observables denoted $X$ and $Z$. 
Only the eigenvectors of the observables will be relevant, call them $\ket{\phi_x}$ and $\ket{\theta_z}$, respectively. 
In infinite dimensions we will confine our analysis to position $Q$ and momentum $P$, taking $\hbar=1$. 
The analog of $Q$ and $P$ in finite dimensions are canonically conjugate observables $X$ and $Z$ for which $\ket{\phi_x}=\tfrac1{\sqrt d}\sum_z \omega^{xz}\ket{\theta_z}$, where $d$ is the dimension and $\omega$ is a primitive $d$th root of unity. 

It will be more convenient for our purposes to adopt the algebraic framework and use the Heisenberg picture, though we shall occasionally employ the Schr\"odinger picture. 
In the Heisenberg picture we describe systems chiefly by the algebra of observables on them and describe transformations of systems by quantum channels, completely positive and unital maps from the algebra of observables of the output to the observables of the input~\cite{davies_quantum_1976,kraus_states_1983,werner_quantum_2001,wolf_quantum_2012,beny_algebraic_2015}.
This allows us to treat classical and quantum systems on an equal footing within the same framework. 
When the input or output system is quantum mechanical, the observables are the bounded operators $\cB(\cH)$ from the Hilbert space $\cH$ associated with the system to itself. 
Classical systems, such as the results of measurement or inputs to a state preparation device, take values in a set, call it $\sY$.
The relevant algebra of observables here is $L^\infty(\sY)$, the (bounded, measureable) functions on $\sY$. 
Hybrid systems are described by tensor products, so an apparatus $\cE$ which measures a quantum system has an output algebra described by $L^\infty(\sY)\otimes \cB(\cH)$. 
To describe just the measurement result, we keep only $L^\infty(\sY)$. 
We shall occasionally denote the input and output spaces explicitly as $\cE_{A\to \sY B}$ when useful. 

For arbitrary input and output algebras $\cA_A$ and $\cA_B$, quantum channels are precisely those maps $\cE$ which are unital, $\cE(\id_B)=\id_A$, and completely positive, meaning that not only does $\cE$ map positive elements of $\cA_B$ to positive elements of $\cA_A$, it also maps positive elements of $\cA_B\otimes \cB(\mathbb C^n)$ to positive elements of $\cA_A\otimes \cB(\mathbb C^n)$ for all integer $n$. 
This requirement is necessary to ensure that channels act properly on entangled systems.

\begin{figure}[h]
\centering
\includegraphics{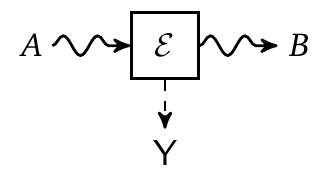}
\caption{\label{fig:apparatus} A general quantum apparatus $\cE$. The apparatus measures a quantum system $A$ giving the output $\textsf{Y}$. In so doing, $\cE$ also transforms the input $A$ into the output system $B$. Here the wavy lines denote quantum systems, the dashed lines classical systems. Formally, the apparatus is described by a quantum instrument.}
\end{figure}

A general measurement apparatus has both classical and quantum outputs, corresponding to the measurement result and the post-measurement quantum system. 
Channels describing such devices are called \emph{quantum instruments}; we will call the channel describing just the measurement outcome a \emph{measurement}. 
In finite dimensions any measurement can be seen as part of a quantum instrument, but not so for idealized position or momentum measurements, as shown in Theorem 3.3 of \cite{davies_quantum_1976} (see page 57).
Technically, we may anticipate the result since the post-measurement state of such a device would presumably be a delta function located at the value of the measurement, which is not an element of $L^2(\mathsf Q)$. 
This need not bother us, though, since it is not operationally meaningful to consider a position measurement instrument of infinite precision. 
And indeed there is no mathematical obstacle to describing finite-precision position measurement by quantum instruments, as shown in Theorem 6.1 (page 67 of \cite{davies_quantum_1976}). 
For any bounded function $\alpha\in L^2(\mathsf Q)$ we can define the instrument $\cE_\alpha:~L^\infty(\sQ)\otimes \cB(\cH)\to \cB(\cH)$ by 
\begin{align}
\label{eq:positioninstrument}
\cE_\alpha(f\otimes a)=\int\!{\rm d}q\, f(q)A_{q;\alpha}^* a A_{q;\alpha}\,,
\end{align}
where $A_{q;\alpha}\psi(q')=\alpha(q-q')\psi(q')$ for all $\psi\in L^2(\sQ)$. 
The classical output of the instrument is essentially the ideal value convolved with the function $\alpha$. 
Thus, setting the width of $\alpha$ sets the precision limit of the instrument. 

\subsection{Distinguishability as a channel norm}
The distinguishability measure is actually a norm on quantum channels, equal (apart from a factor of one half) to the so-called norm of complete boundedness, the cb norm~\cite{kitaev_quantum_1997,paulsen_completely_2003,gilchrist_distance_2005}. 
The cb norm is defined as an extension of the operator norm, similar to the extension of positivity above, as 
\begin{align}
\|T\|_{\rm cb}:=\sup_{n\in \mathbb N}\|\id_n\otimes T\|_\infty\,,
\end{align}
where $\|T\|_\infty$ is the operator norm. 
Then 
\begin{align}
\delta(\cE_1,\cE_2)=\tfrac12\|\cE_1-\cE_2\|_{\rm cb}\,.
\end{align}
In the Schr\"odinger picture we instead extend the trace norm $\|\cdot\|_1$, and the result is usually called the diamond norm~\cite{kitaev_quantum_1997,gilchrist_distance_2005}. 
In either case, the extension serves to account for entangled inputs in the experiment to test whether $\cE_1$ or $\cE_2$ is the actual channel. 
In fact, entanglement is helpful even when the channels describe projective measurements, as shown by an example given in Appendix~\ref{app:entcb}.  
This expression for the cb or diamond norm is not closed-form, as it requires an optimization. 
However, in finite dimensions the cb norm can be cast as a convex optimization, specifically as a semidefinite program~\cite{watrous_semidefinite_2009,watrous_simpler_2013}, which makes numerical computation tractable. 
Further details are given in Appendix~\ref{app:sdp}.

\subsection{The Stinespring representation and its continuity}
\label{sec:stinespring}
According to the Stinespring representation theorem~\cite{stinespring_positive_1955,paulsen_completely_2003}, any channel $\cE$ mapping an algebra $\cA$ to $\cB(\cH)$ can be expressed in terms of an isometry $V:\cH \to \cK$ to some Hilbert space $\cK$ and a representation $\pi$ of $\cA$ in $\cB(\cK)$ such that, for all $a\in\cA$,
\begin{align}
\cE(a)=V^*\pi(a)V\,.
\end{align}
The isometry in the Stinespring representation is usually called the \emph{dilation} of the channel, and $\cK$ the dilation space.
In finite-dimensional settings, calling the input $A$ and the output $B$, one usually considers maps taking $\cA=\cB(\cH_B)$ to $\cB(\cH_A)$. 
Then one can choose $\cK=\cH_B\otimes \cH_E$, where $\cH_E$ is a suitably large Hilbert space associated to the ``environment'' of the transformation ($\cH_E$ can always be chosen to have dimension $\text{dim}(\cH_A)\,\text{dim}(\cH_B)$). 
The representation $\pi$ is just $\pi(a)=a\otimes \id_E$.
Using the isometry $V$, we can also construct a channel from $\cB(\cH_E)$ to $\cB(\cH_A)$ in the same manner; this is known as the complement $\cE^\sharp$ of $\cE$. 

The advantage of the general form of the Stinespring representation is that we can easily describe measurements, possibly continuous-valued, as well. 
For the case of finite outcomes, consider the ideal projective measurement $\cQ_X$ of the observable $X$. 
Choosing a basis $\{\ket{b_x}\}$ of $L^2(\sX)$ and defining $\pi(\delta_x)=\ketbra {b_x}$ for $\delta_x$ the function taking the value 1 at $x$ and zero elsewhere, 
the canonical dilation isometry $W_X:\cH\to L^2(\sX)\otimes \cH$ is given by 
\begin{align}
\label{eq:WX}
W_X=\sum_x \ket{b_x}\otimes \ketbra {\phi_x}\,.
\end{align}
Note that this isometry defines a quantum instrument, since it can describe both the measurement outcome and the post-measurement quantum system. 
If we want to describe just the measurement result, we could simply use $W_X=\sum_x \ket{b_x}\bra{\phi_x}$ with the same $\pi$. 
More generally, a POVM with elements $\Lambda_x$ has the isometry $W_X=\sum_x \ket{b_x}\otimes \sqrt{\Lambda_x}$.

For finite-precision measurements of position or momentum, the form of the quantum instrument in \eqref{eq:positioninstrument} immediately gives a Stinespring dilation $W_Q:\cH\to \cK$ with $\cK=L^2(\sQ)\otimes \cH$ whose action is defined by 
\begin{align}
(W_Q\psi)(q,q')=\alpha(q-q')\psi(q')\,,
\end{align} 
and where $\pi$ is just pointwise multiplication on the $L^\infty(\sQ)$ factor, i.e.\ for $f\in L^\infty(\sQ)$, and $a\in \cB(\cH)$, $[\pi(f\otimes a)(\xi\otimes \psi)](q,q')=f(q)\xi(q)\cdot (a\psi)(q')$ for all $\xi\in L^2(\sQ)$ and $\psi\in \cH$.

A slight change to the isometry in \eqref{eq:WX} gives the dilation of the device which prepares the state $\ket{\varphi_x}$ for classical input $x$. 
Formally the device is described by the map $\cP:\cB(\cH)\to L^2(\sX)$ for which $\cP(\Lambda)=\sum_x \ketbra{b_x} \bra{\varphi_x}\Lambda\ket{\varphi_x}$. 
Now consider $W'_X:L^2(\sX)\to \cH \otimes L^2(\sX)$ given by 
\begin{align}
W'_X=\sum_x \ket{\varphi_x}\otimes \ketbra{b_x}\,.
\end{align}
Choosing $\pi(\Lambda)=\Lambda\otimes \id_{\sX}$, we have $\cP(\Lambda)=W_X'^* \pi(\Lambda)W_X'$.

The Stinespring representation is not unique~\cite{kretschmann_continuity_2008}.  Given two representations $(\pi_1,V_1,\cK_1)$ and $(\pi_2,V_2,\cK_2)$ of the same channel $\cE$, there exists a partial isometry $U:\cK_1\to \cK_2$ such that 
$UV_1=V_2$, $U^*V_2=V_1$, and $U\pi_1(a)=\pi_2(a)U$ for all $a\in \cA$. 
For the representations $\pi$ as usually employed for the finite-dimensional case, this last condition implies that $U$ is a partial isometry from one environment to the other, for $U (a\otimes \id_E)=(a\otimes \id_{E'})U$ can only hold for all $a$ if $U$ acts trivially on $B$.  
For channels describing measurements, finite or continuous, the last condition implies that any such $U$ is a conditional partial isometry, dependent on the outcome of the measurement result.
Thus, for any set of isometries $U_x: \cH_S\to \cH_R$, $\sum_x \ket{b_x}\otimes U_x\ketbra{\phi_x}U_x^*$ is a valid dilation of $\cQ_X$, just as is $W_X$ in \eqref{eq:WX}. 
Similarly, $(W_Q'\psi)(q,q')=\alpha(q-q')[U_q\psi](q')$ is a valid dilation of $\cE_\alpha$ in \eqref{eq:positioninstrument}.

The main technical ingredient required for our results is the continuity of the Stinespring representation in the cb norm~\cite{kretschmann_continuity_2008,kretschmann_information-disturbance_2008}. 
That is, channels which are nearly indistinguishable have Stinespring dilations which are close and vice versa. 
For completely positive and unital maps $\cE_1$ and $\cE_2$, \cite{kretschmann_continuity_2008,kretschmann_information-disturbance_2008} show that
\begin{align}
\label{eq:stinespringcont}
\tfrac12\|\cE_1-\cE_2\|_{\rm cb}\leq \inf_{\pi_i,V_i} \|V_1-V_2\|_\infty \leq \sqrt{\|\cE_1-\cE_2\|_{\rm cb}}\,,
\end{align}
where the infimum is taken over all Stinespring representations $(\pi_i,V_i,\cK_i)$ of $\cE_i$.

\subsection{Sequential and joint measurements}
Using the Stinespring representation we can easily show that, in principle, any joint measurement can always be decomposed into sequential measurement. 
\begin{lemma}
\label{lem:jointseq}
Suppose that $\cE:L^\infty(\sX)\otimes L^\infty(\sZ)\to \cB(\cH)$ is a channel describing a joint measurement. 
Then there exists an apparatus $\cA:L^\infty(\sX)\otimes \cB(\cH')\to \cB(\cH)$ and a conditional measurement $\cM:L^{\!\!\infty}(\sX)\otimes L^\infty(\sZ)\to L^\infty(\sX)\otimes \cB(\cH')$  such that $\cE=\cA\circ \cM$.
\end{lemma}
\begin{proof}
Define $\cM':L^\infty(\sX)\to \cB(\cH)$ to be just the $\sX$ output of $\cE$, i.e. $\cM'(f)=\cE(f\otimes 1)$. 
Now suppose that $V:\cH\to L^2(\sX)\otimes L^2(\sZ)\otimes \cH''$ is a Stinespring representation of $\cE$ and $V_X:\cH\to L^2(\sX)\otimes \cH'$ is a representation of $\cM'$, both with the standard representation $\pi$ of $L^\infty$ into $L^2$.
By construction, $V$ is also a dilation of $\cM'$, and therefore there exists a  partial isometry $U_X$ such that $V=U_XV_X$. 
More specifically, conditional on the value $\sX=x$, each $U_x$ sends $\cH'$ to $L^2(\sZ)\otimes \cH''$.  
Thus, setting $\cA(f\otimes a)=V_X^*(\pi(f)\otimes a)V_X$ and $\cM_x(f)=U_x^*(\pi(f)\otimes \id)U_x$, we have $\cE=\cA\circ \cM$. 
\end{proof}

\section{Definitions of error and disturbance}
\label{sec:definitions}
%\subsection{Figures of merit}
\subsection{Measurement error}
%Now we can give our formal error and disturbance measures. 
To characterize the error $\eps_X$ an apparatus $\cE$ makes relative to an ideal measurement $\cQ_X$ of an observable $X$, we can simply use the distinguishability of the two channels, taking only the classical output of $\cE$. 
Suppose that the apparatus is described by the channel $\cE:\cB(\cH_B)\otimes L^\infty(\sX)\to \cB(\cH_A)$ and the ideal measurement by the channel $\cQ_X:L^\infty(\sX)\to \cB(\cH_A)$. 
To ignore the output system $B$, we make use of the partial trace map $\cT_B: L^\infty(\sX)\to \cB(\cH_B)\otimes L^\infty(\sX)$ given by $\cT_B(f)=\id_B\otimes f$. 
Then a sensible notion of error is given by $\eps_X(\cE)=\delta(\cQ_X,\cE\circ \cT_B)$.
If it is easy to tell the ideal measurement apart from the actual device, then the error is large; if it is difficult, then the error is small. 

As a general definition, though, this quantity is deficient to two respects. 
First, we could imagine an apparatus which performs an ideal $\cQ_X$ measurement, but simply mislabels the outputs. 
This leads to $\eps_X(\cE)=1$, even though the ideal measurement is actually performed. 
Second, we might wish to consider the case that the classical output set of the apparatus is not equal to $\sX$ itself. 
For instance, perhaps $\cE$ delivers much more output than is expected from $\cQ_X$. 
In this case we also formally have $\eps_X(\cE)=1$, since we can just examine the output to distinguish the two devices. 

We can remedy both of these issues by describing the apparatus by the channel $\cE:\cB(\cH_B)\otimes L^\infty(\sY)\to \cB(\cH_A)$ and just including a further classical postprocessing operation $\cR:L^\infty(\sX)\to L^\infty(\sY)$ in the distinguishability step. 
Since we are free to choose the best such map, we define
\begin{align}
\eps_X(\cE):=\inf_{\cR}\, \delta(\cQ_X,\cE\circ \cR\circ \cT_B)\,.
\end{align} 
The setup of the definition is depicted in Figure~\ref{fig:error}.
\begin{figure}[h]
\centering
\includegraphics{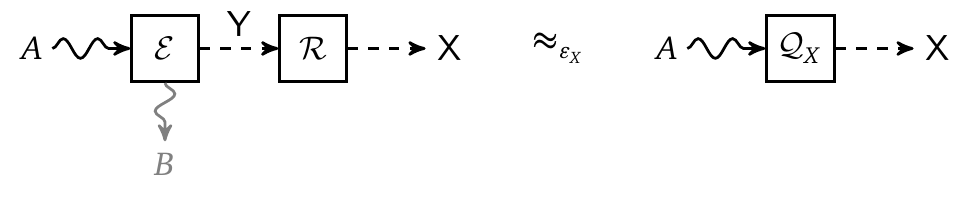}
\caption{\label{fig:error} Measurement error. 
The error made by the apparatus $\cE$ in measuring $X$ is defined by how distinguishable the actual device is from the ideal measurement $\cQ_X$ in any experiment whatsoever, after suitably processing the classical output $\sY$ of $\cE$ with the map $\cR$. 
To enable a fair comparison, we ignore the quantum output of the apparatus, indicated in the diagram by graying out $B$. 
If the actual and ideal devices are difficult to tell apart, the error is small.}
\end{figure}

\subsection{Measurement disturbance}

Defining the disturbance an apparatus $\cE$ causes to an observable, say $Z$, is more delicate, as an observable itself does not have a directly operational meaning. 
But there are two straightforward ways to proceed: we can either associate the observable with measurement or with state preparation.  
In the former, we compare how well we can mimic the ideal measurement $\cQ_Z$ of the observable after employing the apparatus $\cE$, quantifying this using measurement error as before.
Additionally, we should allow the use of recovery operations in which we attempt to ``restore'' the input state as well as possible, possibly conditional on the output of the measurement. 
Formally, let $\cQ_Z:~L^\infty(\sZ)~\to~\cB(\cH_A)$ be the ideal $Z$ measurement and $\cR$ be a recovery map $\cR:\cB(\cH_A)\to \cB(\cH_B)\otimes L^\infty(\sX)$ which acts on the output of $\cE$ conditional on the value of the classical output $\sX$ (which it then promptly forgets). 
As depicted in Figure~\ref{fig:mdist}, the measurement disturbance is then the measurement error after using the best recovery map:
\begin{align}
\label{eq:measdist}
\nu_Z(\cE):=\inf_{\cR}\delta(\cQ_Z,\cE\circ\cR\circ\cT_{\sY}\circ\cQ_Z)\,.
\end{align}

\begin{figure}[h]
\centering
\includegraphics{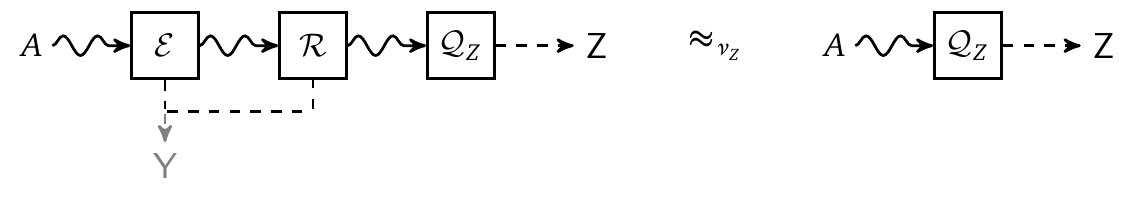}
\caption{\label{fig:mdist} Measurement disturbance. To define the disturbance imparted by an apparatus $\cE$ to the measurement of an observable $Z$, consider performing the ideal $\cQ_Z$ measurement on the output $B$ of $\cE$. First, however, it may be advantageous to ``correct'' or ``recover'' the original input $A$ by some operation $\cR$. In general, $\cR$ may depend on the output $\sX$ of $\cE$. The distinguishability between the resulting combined operation and just performing $\cQ_Z$ on the original input defines the measurement disturbance.}
\end{figure}
%Note that, for projective measurements $\cQ_Z$, we could just subsume the final $\cQ_Z$ into $\cR$. 
%we could just subsume $\cQ_Z$ into $\cR$ for projective measurements... disturbance to information vs to the ``observable''. point is that the system has to be restored, put back to its original configuration --- as in the case where it would then be passed to someone else for the measurement of $\cQ_Z$. 

\subsection{Preparation disturbance}

For state preparation, consider a device with classical input and quantum output that prepares the eigenstates of $Z$. 
We can model this by a channel $\cP_Z$, which in the Schr\"odinger picture produces $\ket{\theta_z}$ upon receiving the input $z$. 
Now we compare the action of $\cP_Z$ to the action of $\cP_Z$ followed by $\cE$, again employing a recovery operation. 
Formally, let $\cP_Z:\cB(\cH_A)\to L^\infty(\sZ)$ be the ideal $Z$ preparation device and consider recovery operations $\cR$ of the form  $\cR:\cB(\cH_A)\to \cB(\cH_B)\otimes L^\infty(\sX)$. 
Then the preparation disturbance is defined as
\begin{align}
\eta_Z(\cE):=\inf_\cR\delta(\cP_Z,\cP_Z\circ\cE\circ\cR\circ\cT_\sY)\,.
\end{align}

\begin{figure}[h]
\centering
\includegraphics{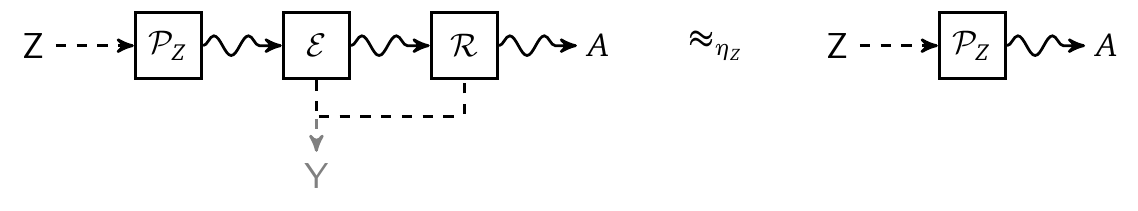}
\caption{\label{fig:pdist} Preparation disturbance. The ideal preparation device $\cP_Z$ takes a classical input $\sZ$ and creates the corresponding $Z$ eigenstate. As with measurement disturbance, the preparation disturbance is related to the distinguishability of the ideal preparation device $\cP_Z$ and $\cP_Z$ followed by the apparatus $\cE$ in question and the best possible recovery operation $\cR$.}
\end{figure}

All of the measures defined so far are ``figures of merit'', in the sense that we compare the actual device to the ideal, perfect functionality. 
In the case of state preparation we can also define a disturbance measure as a ``figure of demerit'', by comparing the actual functionality not to the best-case behavior but to the worst. 
To this end, consider a state preparation device $\cC$ which just ignores the classical input and always prepares the same fixed output state. 
These are constant (output) channels, and clearly $\cE$ disturbs the state preparation $\cP_Z$ considerably if $\cP_Z\circ \cE$ has effectively a constant output. 
Based on this intuition, we can then make the following formal definition:
\begin{align}
\widehat \eta_Z(\cE):=\tfrac{d-1}{d}-\inf_{\cC:\text{const.}}\delta(\cC,\cP_Z\circ \cE)\,. 
\end{align}
The disturbance is small according to this measure if it is easy to distinguish the action of $\cP_Z\circ\cE$ from having a constant output, and large otherwise. 
To see that $\widehat\eta_Z$ is positive, use the Schr\"odinger picture and let the output of $\cC^*$ be the state $\sigma$ for all inputs.
Then note that $\inf_\cC\delta(\cC,\cP_Z\circ\cE)=\min_\cC\max_z\delta(\sigma,\cE^*(\theta_z))$, where the latter $\delta$ is the trace distance. 
Choosing $\sigma=\frac1d\sum_z \cE^*(\theta_z)$ and using joint convexity of the trace distance, we have $\inf_\cC\delta(\cC,\cP_Z\circ\cE) \leq \tfrac{d-1}d$.

\begin{figure}[h]
\centering
\includegraphics{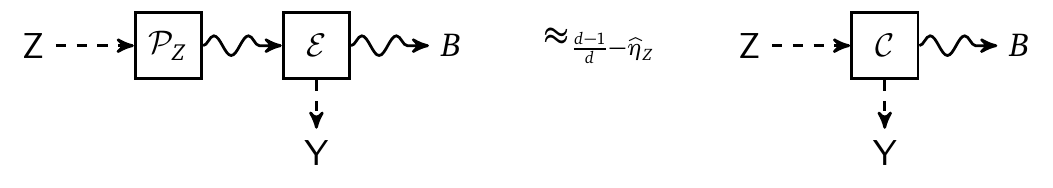}
\caption{\label{fig:pdistdemerit} Figure of ``demerit'' version of preparation  disturbance. Another approach to defining preparation disturbance is to consider distinguishability to a non-ideal device instead of an ideal device. The apparatus $\cE$ imparts a large disturbance to the preparation $\cP_Z$ if the output of the combination $\cP_Z\circ\cE$ is essentially independent of the input. Thus  we  consider the distinguishability of $\cP_Z\circ\cE$ and a constant preparation $\cC$ which outputs a fixed state regardless of the input $\sZ$.}
\end{figure}

We remark that while this disturbance measure leads to finite bounds in the case of finite dimensions, it is less well behaved in the case of position and momentum measurements: Without any bound on the energy of the test states, two channels tend to be as distinguishable as possible, unless they are already constant channels. 
To be more precise, any non-constant channel which only changes the energy by a fixed amount can be differentiated from a constant channel by inputing states of very high energy. 
Roughly speaking, even an arbitrarily strongly disturbing operation can be used to gain \emph{some} information about the input and hence a constant channel is not a good ``worst case'' scenario. 
This is in sharp contrast to the finite-dimensional case, and supports the view that the disturbance measures $\nu_Z(\cE)$ and $\eta_Z(\cE)$ are physically more sensible.

For finite-dimensional systems, all the measures of error and disturbance can be expressed as semidefinite programs, as detailed in Appendix~\ref{app:sdp}.
As an example, we compute these measures for the simple case of a nonideal $X$ measurement on a qubit; we will meet this example later in assessing the tightness of the uncertainty relations and their connection to wave-particle duality relations in the Mach-Zehnder interferometer.
Consider the ideal measurement isometry \eqref{eq:WX}, and suppose that the basis states $\ket{b_x}$ are replaced by two pure states $\ket{\gamma_x}$ which have an overlap $\braket{\gamma_0|\gamma_1}=\sin\theta$. 
Without loss of generality, we can take $\ket{\gamma_x}=\cos\frac\theta 2\ket{b_x}+\sin\frac\theta 2\ket{b_{x+1}}$. 
The optimal measurement $\cQ$ for distinguishing these two states is just projective measurement in the $\ket{b_x}$ basis, so let us consider the channel $\cE_{\text{MZ}}=\cW\circ \cQ$. 
Then, as detailed in Appendix~\ref{app:sdp}, for $Z$ canonically conjugate to $X$ we find
\begin{align}
\eps_X(\cE_{\text{MZ}})&=\tfrac12(1-\cos\theta)\qquad \text{and}\label{eq:MZerror}\\
\nu_Z(\cE_{\text{MZ}})&=\eta_Z(\cE)=\widehat\eta_Z(\cE)=\tfrac12(1-\sin\theta)\,.\label{eq:MZdisturbance}
\end{align}
In all of the figures of merit, the optimal recovery map $\cR$ is to do nothing, while in $\widehat \eta_Z$ the optimal channel $\cC$ outputs the average of the two outputs of $\cP_Z\circ\cE$.

\section{Uncertainty relations in finite dimensions}
\label{sec:finitedimresults}

\subsection{Complementarity measures}
Before turning to the uncertainty relations, we first present several measures of complementarity that will appear therein.  
Indeed, we can use the above notions of disturbance to define several measures of complementarity that will later appear in our uncertainty relations. 
For instance, we can measure the complementarity of two observables just by using the measurement disturbance $\nu$. 
Specifically, treating $\cQ_X$ as the actual measurement and $\cQ_Z$ as the ideal measurement, we define $c_M(X,Z):=\nu_Z(\cQ_X)$. 
This quantity is equivalent to $\eps_Z(\cQ_X)$ since any recovery map $\cR_{\sX\to \sZ}$ in $\eps_Z$ can be used to define $\cR_{\sX\to A}'$ in $\nu_Z$ by $\cR'=\cR\cP_Z$. 
Similarly, we could treat one observable as defining the ideal state preparation device and the other as the measurement apparatus, which leads to $c_P(X,Z):=\eta_Z(\cQ_X)$. 
Here we could also use the ``figure of demerit'' and define $\widehat c_P(X,Z):=\widehat \eta_Z(\cQ_X)$. 

Though the three complementarity measures are conceptually straightforward, it is also desireable to have closed-form expressions, particularly for the bounds in the uncertainty relations. 
To this end, we derive lower bounds as follows. 
First, consider $c_M$ and choose as inputs $Z$ basis states. 
This gives, for random choice of input,
\begin{subequations}
\begin{align}
c_M(X,Z)
&\geq \inf_\cR\,\delta(\cP_Z\cQ_Z,\cP_Z\cQ_X\cR)\\
&\geq 1-\max_{R}\tfrac1d\sum_{xz}  |\langle \varphi_x|\theta_z\rangle|^2 R_{zx}\\
&\geq 1-\max_R \tfrac1d\sum_x \max_z |\langle \varphi_x|\theta_z\rangle|^2 \sum_{z'}R_{z'x}\\
&=1- \tfrac1d\sum_x \max_z |\langle \varphi_x|\theta_z\rangle|^2\,,\label{eq:meritbound0}
\end{align}
\end{subequations}
where the maximization is over stochastic matrices $R$, and we use the fact that $\sum_z R_{zx}=1$ for all $x$. 
For $c_P$ we can proceed similarly. 
Again replacing the recovery map $\cR_{\sX\to A}$ followed by $\cQ_Z$ with a classical postprocessing map $\cR_{\sX\to \sZ}$, we have
\begin{subequations}
\begin{align}
c_P(X,Z)
&\geq \inf_{\cR_{\sX \to A}}\,\delta(\cP_Z\cQ_Z,\cP_Z\cQ_X\cR\cQ_Z)\\
&=\inf_{\cR_{\sX\to \sZ}}\,\delta(\cP_Z\cQ_Z,\cP_Z\cQ_X\cR)\\
&\geq 1-\tfrac1d\sum_x \max_z|\langle \varphi_x|\theta_z\rangle|^2\,.\label{eq:meritbound}
\end{align}
\end{subequations}

For $\widehat c_P(X,Z)$ we have
\begin{subequations}
\begin{align}
\widehat c_P(X,Z)
&=\tfrac{d-1}d-\inf_{\cC:\text{const.}}\delta(\cC,\cP_Z\circ\cQ_X)\\
&=\tfrac{d-1}d -\min_P \max_z \delta(P,\cQ_X^*(\theta_z))\\
&\geq \tfrac{d-1}d-\max_z \tfrac12\sum_x|\tfrac1d-|\langle \varphi_x|\theta_z\rangle|^2|\,,\label{eq:demeritbound}
\end{align}
\end{subequations}
where the bound comes from choosing $P$ to be the uniform distribution. 
We could also choose $P(x)=|\langle \varphi_x|\theta_{z'}\rangle|^2$ for some $z'$ to obtain the bound $\widehat c_P(X,Z)\geq \frac{d-1}d-\min_{z'}\max_z \tfrac12\sum_x\big|\tr[{\varphi_x}(\theta_z-\theta_{z'})]\big|$. 
However, from numerical investigation of random bases, it appears that this bound is rarely better than the previous one. 

Let us comment on the properties of the complementarity measures and their bounds in \eqref{eq:meritbound0}, \eqref{eq:meritbound}, and \eqref{eq:demeritbound}.  
Both expressions in the bounds are, properly, functions only of the two orthonormal bases involved, depending only on the set of overlaps. 
In particular, both are invariant under relabelling the bases. 
Uncertainty relations formulated in terms of conditional entropy typically only involve the largest overlap or largest two overlaps~\cite{coles_entropic_2015,coles_improved_2014}, but the bounds derived here are yet more sensitive to the structure of the overlaps.
Interestingly, the quantity in \eqref{eq:meritbound0} appears in the information exclusion relation of \cite{coles_improved_2014}, where the sum of mutual informations different systems can have about the observables $X$ and $Z$ is bounded by $\log_2 d\sum_x \max_z |\langle \varphi_x|\theta_z\rangle|^2$.

The complementarity measures themselves all take the same value in two extreme cases: zero in the trivial case of identical bases, $(d-1)/d$ in the case that the two bases are conjugate, meaning $|\langle \varphi_x|\theta_z\rangle|^2=1/d$ for all $x,z$. 
In between, however, the separation between the two can be quite large.  
Consider two observables that share two eigenvectors while the remainder are conjugate. 
The bounds \eqref{eq:meritbound0} and \eqref{eq:meritbound} imply that $c_M$ and $c_P$ are both greater than $(d-3)/d$. 
The bound on $\widehat c_P$ from \eqref{eq:demeritbound} is zero, though a better choice of constant channel can easily be found in this case. 
In dimensions $d=3k+2$, fix the constant channel to output the distribution $P$ with probability 1/3 of being either of the last two outputs, $1/3k$ for any $k$ of the remainder, and zero otherwise. 
Then we have $\hat c_P\geq \tfrac{d-1}d-\max_z \delta(P,\cQ_X^*\cP_Z^*(z))$.
It is easy to show the optimal value is $2/3$ so that $\hat c_P\geq(d-3)/3d$. 
Hence, in the limit of large $d$, the gap between the two measures can be at least $2/3$.
This example also shows that the gap between the complementary measures and the bounds can be large, though we will not investigate this further here.   

\subsection{Results}
We finally have all the pieces necessary to formally state our uncertainty relations. The first relates measurement error and measurement disturbance, where we have  
\begin{theorem}
\label{thm:finitemeasdist}
For any two observables $X$ and $Z$ and any quantum instrument $\cE$, 
\begin{align}
\sqrt{2\eps_X(\cE)}+\nu_Z(\cE)&\geq c_M(X,Z) \quad\text{and}\quad \label{eq:epsnu1}\\
\eps_X(\cE)+\sqrt{2\nu_Z(\cE)}&\geq c_M(Z,X)\,.\label{eq:epsnu2}
\end{align}
\end{theorem}
Due to Lemma~\ref{lem:jointseq}, any joint measurement of two observables can be decomposed into a sequential measurement, which implies that these bounds hold for joint measurement devices as well.  
Indeed, we will make use of that lemma to derive \eqref{eq:epsnu2} from \eqref{eq:epsnu1} in the proof below. 
Of course we can replace the $c_M$ quantities with closed-form expressions using the bound in \eqref{eq:meritbound0}.
Figure~\ref{fig:plotmz} shows the bound for the case of conjugate observables of a qubit, for which $c_M(X,Z)=c_M(Z,X)=\tfrac12$. 
It also shows the particular relation between error and measurement disturbance achieved by the apparatus $\cE_{\text{MZ}}$ mentioned at the end of \S\ref{sec:definitions}, from which we can conclude the that bound is tight in the region of vanishing error or vanishing disturbance. 
\begin{figure}
\centering
\includegraphics{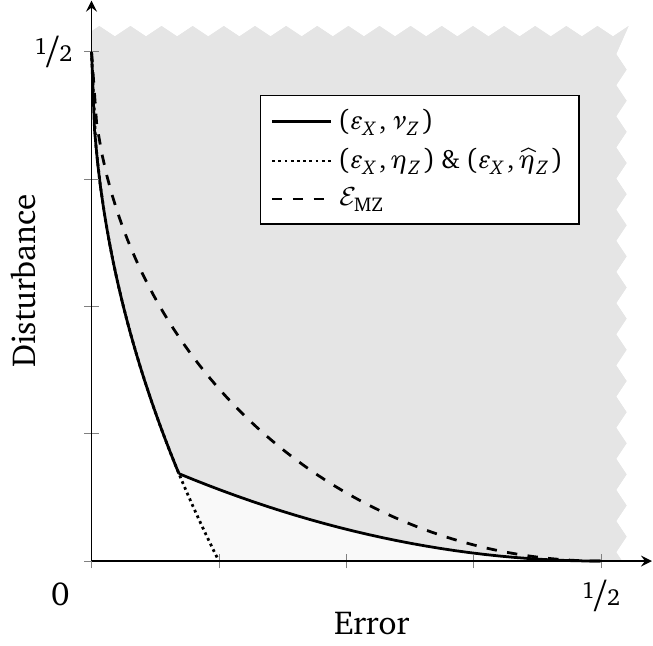}
\caption{\label{fig:plotmz} Error versus disturbance bounds for conjugate qubit observables. 
Theorem~\ref{thm:finitemeasdist} restricts the possible combinations of measurement error $\eps_X$ and measurement disturbance $\nu_Z$ to the dark gray region bounded by the solid line. 
Theorem~\ref{thm:finiteprepdist} additionally includes the light gray region. 
Also shown are the error and disturbance values achieved by $\cE_{\text{MZ}}$ from \S\ref{sec:definitions}.}
\end{figure}

For measurement error and preparation disturbance we find the following relations
\begin{theorem}
\label{thm:finiteprepdist}
For any two observables $X$ and $Z$ and any quantum instrument $\cE$, %, the following bounds holds for any quantum instrument $\cE$,
\begin{align}
\sqrt{2\eps_X(\cE)}+\eta_Z(\cE)&\geq c_P(X,Z) \quad\text{and}\quad \label{eq:epseta1}\\\sqrt{2\eps_X(\cE)}+\widehat\eta_Z(\cE)&\geq \widehat c_P(X,Z)\,.\label{eq:epseta2}
\end{align}
\end{theorem}
Returning to Figure~\ref{fig:plotmz} but replacing the vertical axis with $\eta_Z$ or $\widehat \eta_Z$, we now have only the upper branch of the bound, which continues to the horizontal axis as the dotted line. 
Here we can only conclude that the bounds are tight in the region of vanishing error.

\subsection{Proofs}

The proofs of all three uncertainty relations are just judicious applications of the triangle inequality, and the particular bound comes from the setting in which $\cP_Z$ meets $\cQ_X$. 
We shall make use of the fact that an instrument which has a small error in measuring $\cQ_X$ is close to one which actually employs the instrument associated with $\cQ_X$. This is encapsulated in the following
\begin{lemma}
\label{lem:FandP}
For any apparatus $\cE_{A\to \sY B}$ there exists a channel $\cF_{\sX A\to \sY B}$ such that $\delta(\cE,\cQ'_X\circ \cF)\leq \sqrt{2\eps_X(\cE)}$, where $\cQ'_X$ is a quantum instrument associated with the measurement $\cQ_X$. 
Furthermore, if $\cQ_X$ is a projective measurement, then there exists a state preparation $\cP_{\sX\to \sY B}$ such that $\delta(\cE,\cQ_X\circ \cP)\leq \sqrt{2\eps_X(\cE)}$.
\end{lemma}
\begin{proof}
Let $V:\cH_A\to \cH_B\otimes\cH_E\otimes L^2(\sX)$ and $W_X:\cH_A\to L^2(\sX)\otimes \cH_A$ be respective dilations of $\cE$ and $\cQ_X$. 
Using the dilation $W_X$ we can define the instrument $\cQ'_X$ as
\begin{equation} 
\begin{aligned}
\cQ_X':\,&L^{\!\infty}(\sX)\otimes \cB(\cH_B)\to \cB(\cH_A)\\
&g\otimes A\mapsto W^*_X(\pi(g)\otimes A)W_X\,.
\end{aligned}
\end{equation}
Suppose $\cR_{\sY\to \sX}$ is the optimal map in the definition of $\eps_X(\cE)$, and let $\cR'_{\sY\to \sX\sY}$ be the extension of $\cR$ which keeps the input $\sY$; 
it has a dilation $V':L^2(\sY)\to L^2(\sY)\otimes L^2(\sX)$.
By Stinespring continuity, in finite dimensions there exists a conditional isometry $U_X:L^2(\sX)\otimes \cH_A\to L^2(\sX)\otimes L^2(\sY)\otimes \cH_B\otimes \cH_E$ such that 
\begin{align}
\opnorm{V'V-U_XW_X}\leq \sqrt{2\eps_X(\cE)}\,.
\end{align}

Now consider the map 
\begin{equation} 
\begin{aligned}
\cE':\,&L^{\!\infty}(\sY)\otimes \cB(\cH_B)\to \cB(\cH_A)\\
&f\otimes A\mapsto W^*_XU^*_X(\id_{\sX}\otimes \pi(f)\otimes A\otimes \id_E)U_XW_X\,.
\end{aligned}
\end{equation}
By the other bound in Stinespring continuity we thus have $\delta(\cE,\cE')\leq \sqrt{2\eps_X(\cE)}$.
Furthermore, as described in \S\ref{sec:stinespring}, $U_X$ is a conditional isometry, i.e.\ a collection of isometries $U_x:\cH_A\to L^2(\sY)\otimes \cH_B\otimes \cH_E$ for each measurement outcome $x$. 
Note that we may regard elements of $L^{\infty}(\sX)\otimes \cB(\cH)$ as sequences $(A_x)_{x\in \sX}$ with $A_x\in \cB(\cH)$ for all $x\in \sX$ such that $\text{ess sup}_x \opnorm{A_x}<\infty$. 
Therefore we may define
\begin{equation} 
\begin{aligned}
\cF:\,&L^{\!\infty}(\sY)\otimes \cB(\cH_B)\to L^{\!\infty}(\sX)\otimes \cB(\cH_A)\\
&f\otimes A\mapsto (U^*_x(\pi(f)\otimes A\otimes \id_E)U_x)_{x\in \sX}\,,
\end{aligned}
\end{equation}
so that $\cE'=\cQ_X'\circ \cF$. This completes the proof of the first statement. 

If $\cQ_X$ is a projective measurement, then the output $B$ of $\cQ'_X$ can just as well be prepared from the $\sX$ output. 
Describing this with the map $\cP'_{\sX\to \sX A}$ which prepares states in $A$ given the value of $\sX$ and retains $\sX$ at the output, we have $\cQ'_X=\cQ_X\circ \cP'$.
Setting $\cP=\cP'\circ \cF$ completes the proof of the second statement.  
\end{proof}

Now, to prove \eqref{eq:epsnu1}, start with the triangle inequality and monotonicity. 
Suppose $\cP_{\sX\to \sY B}$ is the state preparation map from Lemma~\ref{lem:FandP}.
Then, for any $\cR_{\sY B\to A}$, 
\begin{subequations}
\label{eq:trianglemono}
\begin{align}
\delta(\cQ_Z,\cQ_X\circ\cP\circ\cR\circ\cQ_Z) 
&\leq \delta(\cQ_Z,\cE\circ\cR\circ\cQ_Z)+\delta(\cE\circ\cR\circ\cQ_Z,\cQ_X\circ\cP\circ\cR\circ\cQ_Z)\\
&\leq \delta(\cQ_Z,\cE\circ\cR\circ\cQ_Z)+\delta(\cE,\cQ_X\circ\cP)\\
&=\delta(\cQ_Z,\cE\circ\cR\circ\cQ_Z)+\sqrt{2\eps_X(\cE)}\,.
\end{align}
\end{subequations}
Observe that $\cP\circ\cR\circ \cQ_Z$ is just a map $\cR'_{\sX\to \sZ}$. 
Taking the infimum over $\cR$ we then have
\begin{subequations}
\begin{align}
\sqrt{2\eps_X(\cE)}+\nu_Z(\cE)
&\geq \inf_\cR \,\delta(\cQ_Z,\cQ_X\circ \cP\circ\cR\circ \cQ_Z)\\
&\geq \inf_\cR \,\delta(\cQ_Z,\cQ_X\circ \cR)\,.
\end{align}
\end{subequations}

To show \eqref{eq:epsnu2}, let $\cR_{\sY B\to A}$ and $\cR'_{\sY\to \sX}$ be the optimal maps in $\nu_Z(\cE)$ and $\eps_X(\cE)$, respectively. 
Now apply Lemma~\ref{lem:jointseq} to $\cM=\cE\cR'\cR\cQ_Z$ and suppose that $\cE'_{A\to \sZ B}$ is the resulting instrument and $\cM_{\sZ B\to \sX}$ is the conditional measurement. 
By the above argument, $\sqrt{2\eps_Z(\cE')}+\nu_X(\cE')\geq \inf_\cR\,\delta(\cQ_X,\cQ_Z\cR)$.  
But $\eps_Z(\cE')\leq \delta(\cQ_Z,\cE'\cT_B)=\nu_Z(\cE)$ and $\nu_X(\cE')\leq \delta(\cQ_X,\cE'\cM)=\eps_X(\cE)$, where in the latter we use the fact that we could always reprepare an $X$ eigenstate and then let $\cQ_X$ measure it. 
Therefore the desired bound holds.\\

To establish \eqref{eq:epseta1}, we proceed just as above to obtain 
\begin{align}
\delta(\cP_Z,\cP_Z\cQ_X\cP\cR)
%&\leq \delta(\cP_Z,\cP_Z\cE\cR)+\delta(\cP_Z\cE\cR,\cP_Z\cQ_X\cP\cR)\\
&\leq \delta(\cP_Z,\cP_Z\cE\cR)+\sqrt{2\eps_X(\cE)}\,.
\end{align}
Now $\cP_{\sX\to \sY B}\cR_{\sY B\to A}$ is a preparation map $\cP_{\sX\to A}$, and taking the infimum over $\cR$ gives 
\begin{subequations}
\begin{align}
\sqrt{2\eps_X(\cE)}+\eta_Z(\cE)
&\geq \inf_{\cR}\,\delta(\cP_Z,\cP_Z\cQ_X\cP\cR)\\
&\geq \inf_{\cP}\,\delta(\cP_Z,\cP_Z\cQ_X\cP)\,.
\end{align}
\end{subequations}

Finally, \eqref{eq:epseta2}. % to \eqref{eq:error-highprepdist}. 
Since the $\widehat \eta_Z$ disturbance measure is defined ``backwards'', we start the triangle inequality with the distinguishability quantity related to disturbance, rather than the eventual constant of the bound. 
For any channel $\cC_{\sZ\to \sX}$ and $\cP_{\sX\to \sY B}$ from Lemma~\ref{lem:FandP}, just as before we have 
\begin{subequations}
\begin{align}
\delta(\cC\cP,\cP_Z\circ\cE)
&\leq \delta(\cC\cP,\cP_Z\circ\cQ_X\cP)+\delta(\cP_Z\circ\cQ_X\cP,\cP_Z\circ\cE)\\
&\leq \delta(\cC,\cP_Z\circ\cQ_X)+\sqrt{2\eps_X(\cE)}\,.
\end{align}
\end{subequations}
Now we take the infimum over constant channels $\cC_{\sZ\to \sX}$. 
Note that
\begin{align} 
\inf_{\cC_{\sZ\to \sY B}}\,\delta(\cC,\cP_Z\cE)\leq \inf_{\cC_{\sZ\to \sX}}\,\delta(\cC\cP,\cP_Z\cE)\,.
\end{align}
Therefore, we have
\begin{align}
\label{eq:breakptdemerit}
\sqrt{2\eps_X(\cE)}+\widehat \eta_Z(\cE)\geq \tfrac{d-1}{d}-\inf_\cC\delta(\cC,\cP_Z\cQ_X)\,.
\end{align}

This last proof also applies to a more general definition of disturbance which does not use $\cP_Z$ at the input, but rather diagonalizes or ``pinches'' any input quantum system in the $Z$ basis. 
Such a transformation can be thought of as the result of performing an ideal $Z$ measurement, but forgetting the result. 
More formally, letting $\cQ_Z^\natural=\cW_Z\circ \cT_\sZ$ with $\cW_Z:a\to W_Z^* aW_Z$, we can define 
\begin{align}
\label{eq:prepdisturbtilde}
\widetilde\eta_Z(\cE)=\tfrac{d-1}{d}-\inf_{\cC}\delta(\cC,\cQ_Z^\natural\circ \cE)\,.
\end{align} 
Though perhaps less conceptually appealing, this is a more general notion of disturbance, since now we can potentially use entanglement at the input to increase distinguishability of $\cQ_Z^\natural\circ \cE$ from any constant channel. 
However, due to the form of $\cQ_Z^\natural$, entanglement will not help.
Applied to any bipartite state, the map $\cQ_Z^\natural$ produces a state of the form $\sum_z p_z \ketbra{\theta_z}\otimes \sigma_z$ for some probability distribution $p_z$ and set of normalized states $\sigma_z$, and therefore the input to $\cE$ itself is again an output of $\cP_Z$. 
Since classical correlation with ancillary systems is already covered in $\widehat\eta_Z(\cE)$, it follows that $\widetilde\eta_Z(\cE)=\widehat\eta_Z(\cE)$.

\section{Position \& momentum}
\label{sec:positionmomentumresults}

\subsection{Gaussian precision-limited measurement and preparation}
Now we turn to the infinite-dimensional case of position and momentum measurements. 
Let us focus on Gaussian limits on precision, where the convolution function $\alpha$ described in \S\ref{sec:systems} is the square root of a normalized Gaussian of width $\sigma$, and for convenience define 
\begin{align}
g_{\sigma}(x)=\frac1{\sqrt{2\pi}\sigma}e^{-\tfrac {x^2}{2\sigma^2}}\,.
\end{align}
One advantage of the Gaussian choice is that the Stinespring dilation of the ideal $\sigma$-limited measurement device is just a canonical transformation. 
Thus, measurement of position $Q$ just amounts to adding this value to an ancillary system which is prepared in a zero-mean Gaussian state with position standard deviation $\sigma_Q$, and similarly for momentum. 
The same interpretation is available for precision-limited state preparation. 
To prepare a momentum state of width $\sigma_P$, we begin with a system in a zero-mean Gaussian state with momentum standard deviation $\sigma_P$ and simply shift the momentum by the desired amount. 
%Let us denote the ideal measurement instruments $\cQ'_Q$ and $\cQ'_P$, which are associated with precisions $\sigma_Q$ and $\sigma_P$, respectively. 
%Similarly, $\cP_Q$ and $\cP_P$ are the ideal position and momentum preparation devices, with the same precisions. 

Given the ideal devices, the definitions of error and disturbance are those of \S\ref{sec:definitions}, as in the finite-dimensional case, with the slight change that the first term of $\widehat \eta$ is now 1.
To reduce clutter, we do not indicate $\sigma_Q$ and $\sigma_P$ specifically in the error and disturbance functions themselves. %shall associate precision $\sigma_Q$ with position and $\sigma_P$ with momentum, but not indicate these quantities specifically in the error and disturbance functions themselves. 

%To reduce clutter, we shall associate width $\sigma_Q$ with position and width $\sigma_P$ with momentum, but not indicate these quantities specifically in the error and disturbance themselves. 
Since our error and disturbance measures are based on possible state preparations and measurements in order to best distinguish the two devices, in principle one ought to consider precision limits in the distinguishability quantity $\delta$. 
However, we will not follow this approach here, and instead we allow test of arbitrary precision in order to preserve the link between distinguishability and the cb norm. 
This leads to bounds that are perhaps overly pessimistic, but nevertheless limit the possible performance of any device. 

\subsection{Results}
As discussed previously, the disturbance measure of demerit $\widehat \eta$ cannot be expected to lead to uncertainty relations for position and momentum observables, as any non-constant channel can be perfectly differentiated from a constant one by inputting states of arbitrarily high momentum. We thus focus on the disturbance measures of merit.

\begin{theorem}
\label{thm:infinitemerit}
Set $c=2\sigma_Q\sigma_P$ for any precision values $\sigma_Q,\sigma_P>0$. 
Then for any quantum instrument $\cE$,
\begin{align}
\left.
 \begin{array}{c}
\sqrt{2\eps_{Q}(\cE)}+\nu_P(\cE)\\[2mm] \eps_Q(\cE)+\sqrt{2\nu_{Q}(\cE)}\end{array}\right\} &\geq \frac{1-c^2}{(1+c^{2/3}+c^{4/3})^{3/2}}\quad \text{and}\label{eq:infiniteerror-lowmeasdist}\\
\sqrt{2\eps_{Q}(\cE)}+\eta_P(\cE) 
&\geq \frac{(1+c^2)^{1/2}}{((1+c^2)+c^{2/3}(1+c^2)^{2/3}+c^{4/3}(1+c^2)^{1/3})^{3/2}}\,.\label{eq:infiniteerror-lowprepdist}
\end{align}
\end{theorem}

Before proceeding to the proofs, let us comment on the properties of the two bounds.
As can be seen in Figure~\ref{fig:plotbounds}, the bounds take essentially the same values for $\sigma_Q\sigma_P\ll \tfrac12$, and indeed both evaluate to unity at $\sigma_Q\sigma_P=0$. 
This is the region of combined position and momentum precision far smaller than the natural scale set by $\hbar$, and the limit of infinite precision accords with the finite-dimensional bounds for conjugate observables in the limit $d\to \infty$. 
Otherwise, though, the bounds differ remarkably. 
The measurement disturbance bound in \eqref{eq:infiniteerror-lowmeasdist} is positive only when $\sigma_Q\sigma_P\leq \tfrac 12$, which is the Heisenberg precision limit. 
In contrast, the preparation disturbance bound in \eqref{eq:infiniteerror-lowprepdist} is always positive, though it decays roughly as $(\sigma_Q\sigma_P)^2$.

The distinction between these two cases is a result of allowing arbitrarily precise measurements in the distinguishability measure. 
It can be understood by the following heuristic argument. 
%rewrite to say: even if we input an arbitrarily precise momentum state, the position measurement will kick the momentum by 1/\sigma_Q and therefore we cannot detect the change in momentum for resolutions above 1/\sigma_Q.
Consider an experiment in which a momentum state of width $\sigma_P^{\text{in}}$ is subjected to a position measurement of resolution $\sigma_Q$ and then a momentum measurement of resolution $\sigma_P^{\text{out}}$. 
From the uncertainty principle, we expect the position measurement to change the momentum by an amount $\sim 1/\sigma_Q$. 
Thus, to reliably detect the change in momentum, $\sigma_P^{\text{out}}$ must fulfill the condition $\sigma_P^{\text{out}}\ll \sigma_{P}^{\text{in}}+1/\sigma_Q$. 
The Heisenberg limit in the measurement disturbance scenario is $\sigma_P^{\text{out}}=2/\sigma_Q$, meaning this condition cannot be met no matter how small we choose $\sigma_P^{\text {in}}$. 
%When $\sigma_Q\sigma_P\geq \tfrac12$, the outcomes of $\sigma_P$-limited momentum measurements will presumably hardly differ whether $\sigma_Q$-limited position measurement is previously performed or not. 
This is consistent with no nontrivial bound in \eqref{eq:infiniteerror-lowmeasdist} in this region.  
On the other hand, for preparation disturbance the Heisenberg limit is $\sigma_P^{\text{in}}=2/\sigma_Q$, so detecting the change in momentum simply requires $\sigma_P^{\text{out}}\ll 1/\sigma_Q$. 
A more satisfying approach would be to include the precision limitation in the distinguishability measure to restore the symmetry of the two scenarios, but this requires significant changes to the proof and is left for future work.

\begin{figure}[h]
\centering
\includegraphics{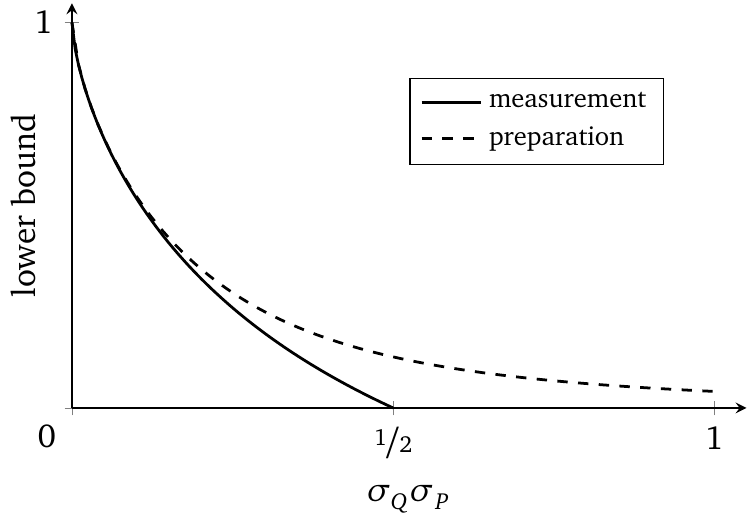}
\caption{\label{fig:plotbounds} Uncertainty bounds appearing in Theorem~\ref{thm:infinitemerit} in terms of the combined precision $\sigma_Q\sigma_P$. The solid line corresponds to the bound involving measurement disturbance, \eqref{eq:infiniteerror-lowmeasdist}, the dashed line to the bound involving preparation disturbance, \eqref{eq:infiniteerror-lowprepdist}.}
\end{figure}

\subsection{Proofs}
The proof of Theorem~\ref{thm:infinitemerit} is broadly similar to the  finite-dimensional case. 
We would again like to begin with $\cF_{\sQ A\to \sY B}$ from Lemma~\ref{lem:FandP} such that $\delta(\cE,\cQ'_Q\cF)\leq \sqrt{2\eps_Q(\cE)}$. 
However, the argument does not quite go through, as in infinite dimensions we cannot immediately ensure that the infimum in Stinespring continuity is attained. 
Nonetheless, we can consider a sequence of maps $(\cF_n)_{n\in \mathbb N}$ such that the desired distinguishability bound holds in the limit $n\to \infty$. 

To show \eqref{eq:infiniteerror-lowmeasdist}, we follow the steps in \eqref{eq:trianglemono}. 
Now, though, consider the map $\cF_n'$ which just appends $\sQ$ to the output of $\cF_n$, and define $\cN=\cQ_Q'\cF_n\cR\cQ_P$, where $\cQ_Q'$ is the instrument associated with position measurement $\cQ_Q$.
Then we have
\begin{subequations}
\begin{align}
\delta(\cQ_P,\cN\cT_\sQ) 
&\leq \delta(\cQ_P,\cE\circ\cR\circ\cQ_P)+\delta(\cE\circ\cR\circ\cQ_P,\cN\cT_\sQ)\\
&\leq \delta(\cQ_P,\cE\circ\cR\circ\cQ_P)+\delta(\cE,\cQ_Q'\circ\cF_n)\,.
\end{align}
\end{subequations}
Taking the limit $n\to \infty$ and the infimum over recovery maps $\cR$ produces $\sqrt{2\eps_Q(\cE)}+\nu_P(\cE)$ on the righthand side. 
We can bound the lefthand side by testing with pure unentangled inputs:
%Defining $\cN=\cQ_Q'\cF_n\cR\cQ_P$ for convenience, we have
\begin{align}
\label{eq:Nlower}
\delta(\cQ_P,\cN\cT_\sQ)
&\geq \sup_{\psi,f}\braket{\psi, \left(\cQ_P(f) - [\cN\cT_\sQ](f) \right) \psi}.
\end{align}

Now we want to show that, since $\cQ_P$ is covariant with respect to phase space translations, without loss of generality we can take $\cN$ to be covariant as well. Consider the translated version of both $\cQ_P$ and $\cN\cT_\sQ$, obtained by shifting their inputs and outputs correspondingly by some amount $z=(q,p)$. 
For the states $\psi$ this shift is implemented by the Weyl-Heisenberg operators $V_z$, while for tests $f$ only the value of $p$ is relevant. 
Any such shift does not change the distinguishability, because we can always shift $\psi$ and $f$ as well to recover the original quantity. 
Averaging over the translated versions therefore also leads to the same distinguishability, and since $\cQ_P$ is itself covariant, the averaging results in a covariant $\cN\cT_\sQ$. 
The details of the averaging require some care in this noncompact setting, but are standard by now, and we refer the reader to the work of Werner~\cite{werner_uncertainty_2004} for furter details. 
Since $\cT_\sQ$ just ignores the $\sQ$ output of the measurement $\cN$, we may thus proceed by assuming that $\cN$ is a covariant measurement.

Any covariant $\cN$ has the form 
\begin{align}
\cN(f)=\int_{\mathbb R^2}\frac{\text{d}z}{{2\pi}}\,f(z) V_{z}mV_{z}^*,
\end{align}
for some positive operator $m$ such that $\tr[m]=1$. 
Due to the definition of $\cN$, the position measurement result is precisely that obtained from $\cQ_Q$. 
By the covariant form of $\cN$, this implies that the position width of $m$ is just $\sigma_Q$ (or rather that of the parity version of $m$, see~\cite{werner_uncertainty_2004}). 
Suppose the momentum distribution has standard deviation $\widehat \sigma_P$; then $\sigma_Q\widehat \sigma_P\geq 1/2$ follows from the Kennard uncertainty relation~\cite{kennard_zur_1927}. 

Now we can evaluate the lower bound term by term. 
Let us choose a Gaussian state in the momentum representation and test function: $\psi=g_{\sigma_\psi}^\half$ and $f=\sqrt{2\pi}\sigma_f g_{\sigma_f}$.
Then the first term is a straightforward Gaussian integral, since the precision-limited measurement just amounts to the ideal measurement convolved with $g_{\sigma_P}$:
\begin{subequations}
\begin{align}
    \braket{\psi, \cQ_P(f)\psi} 
    %&= \braket{\psi, W_P^* \pi\left( f\right) W_P \psi} \\ %&= \int dx dx' dx'' \; \sqrt{g_{\sigma_\text{X}}\left(x-x'\right)}\sqrt{g_{\sigma_\text{X}}\left(x-x''\right)} \braket{\varphi, P_{x'}P_{x''}\varphi} f(x)\\
    &= \int_{\mathbb R^2} dp' dp\; g_{\sigma_\psi}(p')g_{\sigma_P}(p'-p) f(p) \\
    %&= \int dx\; g_{\sqrt{\sigma_X^2 + \sigma_\varphi^2 }}\left(x\right) f(x) \\
    %&= \frac{1}{\sqrt{2\pi\left(\sigma_X^2 + \sigma_\psi^2\right)}}\int dx\; \exp\left[-\frac{x^2}{2}\left(\frac{1}{\sigma_X^2 + \sigma_\psi^2}+ \frac{1}{\sigma_f^2}\right)\right]\\
    %&= \frac{1}{\sqrt{2\pi\left(\sigma_X^2 + \sigma_\psi^2\right)}}\sqrt{2\pi} \left( \sqrt{\frac{1}{\sigma_X^2 + \sigma_\psi^2} + \frac{1}{\sigma_f^2}}\right)^{-1}\\
    %&= \frac{1}{\sqrt{1 + \frac{\sigma_X^2 + \sigma_\psi^2}{\sigma_f^2}}} \\
    &= \frac{\sigma_f}{\sqrt{\sigma_f^2 + \sigma_P^2 + \sigma_\psi^2}}.
    \label{eq:jmone}
\end{align}
\end{subequations}
The second term is the same, just with $\widehat\sigma_P$ instead of $\sigma_P$, so we have
\begin{align}
\label{eq:uglydiff}
\delta(\cQ_{P},\cN\cT_\sQ) 
&\geq \frac{\sigma_f}{\sqrt{\sigma_f^2 + \sigma_P^2 + \sigma_\psi^2}}-\frac{\sigma_f}{\sqrt{\sigma_f^2 + \widehat\sigma_P^2 + \sigma_\psi^2}}.
\end{align}
The tightest possible bound comes from the smallest $\widehat \sigma_P$, which is $1/2\sigma_Q$, and the bound is clearly trivial if $\sigma_Q\sigma_P\geq 1/2$. 
%\jmr{(What can we actually achieve with covariant measurement?)}
If this is not the case, we can optimize our choice of $\sigma_f$. 
To simplify the calculation, assume that $\sigma_\psi$ is small compared to $\sigma_f$ (so that we are testing with a very narrow momentum state).
Then, with $c=2\sigma_Q\sigma_P$, the optimal $\sigma_f$ is given by  
\begin{equation}
  \sigma_f^2 = \frac{\sigma_P^2}{c^{2/3}(1+c^{2/3})}.
  \label{eq:optimalsf}
\end{equation}
Using this in \eqref{eq:uglydiff} gives \eqref{eq:infiniteerror-lowmeasdist}.

For preparation disturbance, proceed as before to obtain %let $\cP'=\cP_P\cQ_Q'\cF_n'\cR$ to obtain
\begin{subequations}
\begin{align}
\delta(\cP_P,\cP_P\cQ_Q'\cF_n'\cR\cT_\sQ)
&\leq \delta(\cP_P,\cP_P\cE\cR)+\delta(\cP_P\cE\cR,\cP_P\cQ_Q'\cF_n'\cR\cT_\sQ)\\
&\leq  \delta(\cP_P,\cP_P\cE\cR)+\delta(\cE,\cQ_Q'\cF_n)
%&\leq \delta(\cP_Z,\cP_Z\cE\cR)+\delta(\cP_Z\cE\cR,\cP_Z\cQ_X\cP\cR)\\
%&\leq \delta(\cP_Z,\cP_Z\cE\cR)+\sqrt{2\eps_X(\cE)}\,.
\end{align}
\end{subequations}
Now the limit $n\to \infty$ and the infimum over recovery maps $\cR$ produces $\sqrt{2\eps_Q(\cE)}+\eta_P(\cE)$ on the righthand side. 
A lower bound on the quantity on the lefthand side can be obtained by using $\cP_P$ to prepare a $\sigma_P$-limited input state and making a $\sigma_m$-limited momentum measurement $\bar\cQ_P$ measurement on the output, so that, for $\cN$ as before,
\begin{align}
\delta(\cP_P,\cP_P\cQ_Q'\cF_n'\cR\cT_\sQ) \geq \sup_{\psi:\text{Gaussian};f} \braket{\psi, \left(\bar\cQ_P(f) - [\cN\cT_\sQ](f) \right) \psi}.
\end{align}
The only difference to \eqref{eq:Nlower} is that the supremum is restricted to Gaussian states of width $\sigma_P$. 
The covariance argument nonetheless goes through as before, and we can proceed to evaluate the lower bound as above. 
This yields
\begin{align}
\delta(\cP_P,\cP_P\cQ_Q'\cF_n'\cR\cT_\sQ) \geq
\frac{\sigma_f}{\sqrt{\sigma_f^2 + \sigma_m^2 + \sigma_P^2}}-\frac{\sigma_f}{\sqrt{\sigma_f^2 + \frac1{4\sigma_Q^2} + \sigma_P^2}}\,.
\end{align}
We may as well consider $\sigma_m\to 0$ so as to increase the first term. 
The optimal $\sigma_f$ is then given by the optimizer above, replacing $c$ with $c/\sqrt{1+c^2}$. 
Making the same replacement in \eqref{eq:infiniteerror-lowmeasdist} yields \eqref{eq:infiniteerror-lowprepdist}. 

\section{Applications}
\label{sec:app}
\subsection{No information about $Z$ without disturbance to $X$}
\label{sec:appinfodisturb}
A useful tool in the construction of quantum information processing protocols is the link between reliable transmission of $X$ eigenstates through a channel $\cN$ and $Z$ eigenstates through its complement $\cN^\sharp$, particularly when the observables $X$ and $Z$ are maximally complementary, i.e.\ $|\langle \varphi_x|\vartheta_z\rangle|^2=\frac 1d$ for all $x,z$. Due to the uncertainty principle, we expect that a channel cannot reliably transmit the bases to different outputs, since this would provide a means to simultaneously measure $X$ and $Z$. 
This link has been used by Shor and Preskill to prove the security of quantum key distribution~\cite{shor_simple_2000} and by Devetak to determine the quantum channel capacity~\cite{devetak_private_2005}.
Entropic state-preparation uncertainty relations from \cite{berta_uncertainty_2010,tomamichel_uncertainty_2011} can be used to understand both results, as shown in~\cite{renes_duality_2011,renes_physics_2012}. 

However, the above approach has the serious drawback that it can only be used in cases where the specific $X$-basis transmission over $\cN$ and $Z$-basis transmission over $\cN^\sharp$ are in some sense compatible and not \emph{counterfactual}; because the argument relies on a state-dependent uncertainty principle, both scenarios must be compatible with the same quantum state. 
Fortunately, this can be done for both QKD security and quantum capacity, because at issue is whether $X$-basis ($Z$-basis) transmission is reliable (unreliable) on average \emph{when the states are selected uniformly at random}. Choosing among either basis states at random is compatible with a random measurement in either basis of half of a maximally entangled state, and so both $X$ and $Z$ basis scenarios are indeed compatible. The same restriction to choosing input states uniformly appears in the recent result of \cite{buscemi_noise_2014}, as it also ultimately relies on a state-preparation uncertainty relation.

Using Theorem~\ref{thm:finiteprepdist} we can extend the method above to counterfactual uses of arbitrary channels $\cN$, in the following sense: If acting with the channel $\cN$ does not substantially affect the possibility of performing an $X$ measurement, then $Z$-basis inputs to $\cN^\sharp$ result in an essentially constant output. More concretely, we have 
\begin{corollary}
\label{cor:leakage}
  Given a channel $\cN$ and complementary channel $\cN^\sharp$, suppose that there exists a measurement $\Lambda_X$ such that $\delta(\cQ_X,\cN\circ \Lambda_X)\leq \eps$. 
  Then there exists a constant channel $\cC$ such that 
  \begin{align}
  \delta(\cQ_Z^{\natural}\circ \cN^\sharp, \cC)\leq \sqrt{2\eps}+\tfrac{d-1}d-\widehat c_P(X,Z).
  \end{align} 
  For maximally complementary $X$ and $Z$, $\delta(\cQ_Z^{\natural}\circ\cN^\sharp,\cC)\leq \sqrt{2\eps}$.
\end{corollary}
\begin{proof}
Let $V$ be the Stinespring dilation of $\cN$ such that $\cN^\sharp$ is the complementary channel and define $\cE={\mathcal V}_\cN\circ\Lambda_X$. 
For $\cC$ the optimal choice in the definition of $\widehat \eta_Z(\cE)$, \eqref{eq:epseta2}, \eqref{eq:prepdisturbtilde}, and $\widetilde \eta_Z=\widehat\eta_Z$ imply $\delta(\cQ_Z^\natural\circ\cE,\cC)\leq \sqrt {2\eps}+\frac{d-1}d-\widehat c_P(X,Z)$. Since $\cN^\sharp$ is obtained from $\cE$ by ignoring the $\Lambda_X$ measurement result, $\delta( \cQ_Z^\natural\circ \cN^\sharp,\cC)\leq \delta( \cQ_Z^\natural\circ \cE,\cC)$.
\end{proof}

This formulation is important because in more general cryptographic and communication scenarios we are interested in the worst-case behavior of the protocol, not the average case under some particular probability distribution. For instance, in~\cite{lacerda_classical_2014} the goal is to construct a classical computer resilient to leakage of $Z$-basis information by establishing that reliable $X$ basis measurement is possible despite the interference of the eavesdropper. However, such an $X$ measurement is entirely counterfactual and cannot be reconciled with the actual $Z$-basis usage, as the $Z$-basis states will be chosen \emph{deterministically} in the classical computer. 

It is important to point out that, unfortunately, calibration testing is in general completely insufficient to establish a small value of $\delta(\cQ_X,\cN\circ \Lambda_X)$. 
More specifically, the following example shows that there is no dimension-independent bound connecting $\inf_{\Lambda_X}\delta(\cQ_X,\cN\circ \Lambda_X)$ to the worst case probability of incorrectly identifying an $X$ eigenstate input to $\cN$, for arbitrary $\cN$.
Let the quantities $p_{yz}$ be given by $p_{y,0}=2/d$ for $y=0,\dots, d/2-1$, $p_{y,1}=2/d$ for $y=d/2,\dots,d-1$, and $p_{y,z}=1/d$ otherwise, where we assume $d$ is even, and then define the isometry $V:\cH_A\to \cH_B\otimes \cH_C\otimes\cH_D$ as the map taking  $\ket z_A$ to $\sum_y \sqrt{p_{yz}}\ket y_B\ket z_C\ket y_D$.
Finally, let $\cN:\cB(\cH_B)\otimes \cB(\cH_C)\to \cB(\cH_A)$ be the channel obtained by ignoring $D$, i.e.\ in the Schr\"odinger picture $\cN^*(\rho)=\tr_D[V\rho V^*]$.
Now consider inputs in the $X$ basis, with $X$ canonically conjugate to $Z$.
As shown in Appendix~\ref{app:counterexample}, the probability of correctly determining any particular $X$ input is the same for all values, and is equal to $\frac1{d^2}\sum_y \left(\sum_z \sqrt{p_{y,z}}\right)^2=(d+\sqrt 2-2)^2/d^2$.
The worst case $X$ error probability therefore tends to zero like $1/d$ as $d\to \infty$.
On the other hand, $Z$-basis inputs $0$ and $1$ to the complementary channel $\cE^\sharp$ result in completely disjoint output states due to the form of $p_{yz}$. 
Thus, if we consider a test which inputs one of these randomly and checks for agreement at the output, we find $\inf_\cC\delta(\cQ_Z^\natural\circ\cN^\sharp,\cC)\geq \tfrac12$. Using the bound above, this implies $\inf_{\Lambda_X}\delta(\cQ_X,\cN\circ\Lambda_X)\geq \frac 18$.
This is not 1, but the point is it is bounded away from zero and independent of $d$: There must be a factor of $d$ when converting between the worst case error probability and the distinguishability.

We can appreciate the failure of calibration in this example from a different point of view, by appealing to the information-disturbance tradeoff of~\cite{kretschmann_information-disturbance_2008}.
Since $\cN$ transmits $Z$ eigenstates perfectly to $BC$ and $X$ eigenstates almost perfectly, we might be tempted to conclude that the channel is close to the identity channel. 
However, the information-disturbance tradeoff implies that complements of channels close to the identity are close to constant channels. 
Clearly this is not the case here, since $\mathcal N^*(\ketbra 0)$ is distinguishable from $\mathcal N^*(\ketbra 1)$. 
This point is discussed further by one of us in~\cite{renes_uncertainty_2016-1}.
The counterexample constructed above it not symmetric for $Z$ inputs, and it is an open question if calibration is sufficient in the symmetric case. 
For channels that are covariant with respect to the Weyl-Heisenberg group (also known as the generalized Pauli group), it is not hard to show that calibration \emph{is} in fact sufficient.

\subsection{Connection to wave-particle duality relations}
In~\cite{englert_fringe_1996} Englert presents a wave-particle complementarity relation in a Mach-Zehnder interferometer, quantifying the extent to which ``the observations of an interference pattern and the acquisition of which-way information are mutually exclusive''. 
The particle-like ``which-way'' information is obtained by additional detectors in the arms of the interferometer, while fringe visibility is measured by the population difference between the two output ports of the interferometer. 
The detectors can be thought of as producing different states in an ancilla system, depending on the path taken by the light. 
Englert shows the following tradeoff between the visibility $\cV$ and distinguishability $\cD$ of the which-way detector states:
\begin{align}
\label{eq:englert}
\cV^2+\cD^2\leq 1.  
\end{align}

We may regard the entire interferometer plus which-way detector as an apparatus $\cE_{\rm MZ}$ with quantum and classical output. 
It turns out that $\cE_{\text{MZ}}$ is precisely the nonideal qubit $X$ measurement considered in \S\ref{sec:definitions} and that path distinguishability is related error of $X$ and visibility to disturbance (all of which are equal in this case by \eqref{eq:MZdisturbance}) of a conjugate observable $Z$.  
More specifically, as shown in Appendix~\ref{app:englert}, 
\begin{align}
\label{eq:comprelations}
\eps_X(\cE_{\rm MZ}) &=\tfrac12(1-\cD)\quad\text{and}\quad
\nu_Z(\cE_{\rm MZ})=\eta_Z(\cE_{\rm MZ})=\widehat \eta_Z(\cE_{\rm MZ}) = \tfrac12(1-\cV).
\end{align} 
Therefore, \eqref{eq:englert} is also an error-preparation disturbance relation. 
By the same token, the uncertainty relations in Theorems~\ref{thm:finitemeasdist} and \ref{thm:finiteprepdist} imply wave-particle duality relations. 

Let us comment on other connections between uncertainty and duality relations. 
Recently, \cite{coles_equivalence_2014} showed a relation between wave-particle duality relations and entropic uncertainty relations. 
As discussed above, the latter are state-dependent state-preparation relations, and so the interpretation of the wave-particle duality relation is somewhat different. 
Here we have shown that Englert's relation can actually be understood as a state-independent relation.

Each of the disturbance measures are related to visibility in Englert's setup.
It is an interesting question to consider a multipath interferometer to settle the question of which disturbance measure should be associated to visibility in general. From the discussion of \cite{coles_entropic_2016}, it would appear that visibility ought to be related to measurement disturbance $\nu_Z$, but we leave a complete analysis to future work.

\section{Comparison to previous work}
\label{sec:comparison}

Broadly speaking, there are two main kinds of uncertainty relations: those which are constraints on fixed experiments, including the details of the input quantum state, and those that are constraints on quantum devices themselves, independent of the particular input. 
All of our relations are of the latter type, in contrast to entropic relations, which are typically of the former type.
At a formal level, this distinction appears in whether or not the quantities involved in the precise relation depend on the input state or not.\footnote{This is separate from the issue of whether the bound depends on the state, as for instance in the Robertson relation~\cite{robertson_uncertainty_1929}.%entropic bound with quantum memory~\cite{berta_uncertainty_2010}.
} 
Each type of relation certainly has its use, though when considering error-disturbance uncertainty relations, we argued in the introduction that the conceptual underpinnings of state-dependent relations describing fixed experiments are unclear.
Indeed, it is precisely because of the uncertainty principle that trouble arises in defining error and disturbance in this case.  %since it becomes problematic not only to define error and disturbance but also to relate these measures in the context of a single experiment.
Worse still, there can be no nontrivial bound relating error and disturbance which applies universally, i.e.\ to all states~\cite{korzekwa_operational_2014}. 

Independent of the previous question, another major contrast between different kinds of uncertainty relations is whether they depend on the values taken by the observables, or only the configuration of their eigenstates. 
Again, our relations are all of the latter type, but now we share this property with entropic relations. 
That is not to say that the observable values are completely irrelevant in our setting, merely that they are not necessarily relevant.
In distinguishing the outputs of an ideal position measurement of given precision from the outputs of the actual device, one may indeed make use of the difference in measurement values. 
But this need not be the only kind of comparison.

In the recent work of Busch, Lahti, and Werner~\cite{busch_proof_2013}, the authors used the Wasserstein metric of order two, corresponding to the mean squared error, as the underlying distance $D(.,.)$ to measure the closeness of probability distributions. 
If $\cM^Q$, $\cM^P$ are the marginals of some joint measurement of position $Q$ and momentum $P$, and $X_\rho$ denotes the distribution coming from applying the measurement $X$ to the state $\rho$, their relation reads
\begin{align}
  \label{eq:relation_wernergang}
  \sup_\rho D(\cM^Q_\rho,Q_\rho) \cdot \sup_\rho D(\cM^P_\rho,P_\rho) \geq c \,,
\end{align}
for some universal constant $c$. In~\cite{busch_measurement_2014}, the authors generalize their results to arbitrary Wasserstein metrics. 
As in our case, the two distinguishability quantities in \eqref{eq:relation_wernergang} are separately maximized over all states, and hence the resulting expression characterizes the goodness of the approximate measurement. 

One could instead ask for a ``coupled optimization'', a relation of the form 
\begin{align}
\label{eq:coupledsup}
\sup_\rho \left[D(\cM^Q_\rho,Q_\rho) D(\cM^P_\rho,P_\rho) \right]\geq c',
\end{align}
for some other constant $c'$.\footnote{Such an approach has been advocated by David Reeb (private communication).} 
This approach is taken in \cite{barchielli_measurement_2016} for the question of joint measurability. 
While this statement certainly tells us that no device can accurately measure both position and momentum for all input states, the bound $c'$ only holds (and can only hold) for the worst possible input state. 
In contrast, our bounds, as well as in \eqref{eq:relation_wernergang} are {state-independent} in the sense that the bound holds for all states. 
Indeed, the two approaches are more distinct than the similarities between \eqref{eq:relation_wernergang} and \eqref{eq:coupledsup} would suggest. 
By optimizing over input states separately, our results and those of \cite{werner_uncertainty_2004,busch_proof_2013,busch_measurement_2014} are statements about the properties of measurement devices themselves, independent of any particular experimental setup. 
State-dependent settings capture the behavior of measurement devices in specific experimental setups and must therefore account for the details of the input state. 

The same set of authors also studied the case of finite-dimensional systems, in particular qubit systems, again using the Wasserstein metric of order two~\cite{busch_heisenberg_2014}. 
Their results for this case are similar, with the product in \eqref{eq:relation_wernergang} replaced by a sum. 
Perhaps most closely related to our results is the recent work by Ipsen~\cite{ipsen_error-disturbance_2013}, who uses the variational distance as the underlying distinguishability measure to derive similar additive uncertainty relations. 
We note, however, that both \cite{busch_heisenberg_2014} and \cite{ipsen_error-disturbance_2013} only consider joint measurability and do not consider the change to the state after the approximate measurement is performed, as it is done in our error-disturbance relation. 
Furthermore, both base their distinguishability measures on the measurement statistics of the devices alone. 
But this does not necessarily tell us how distinguishable two devices ultimately are, as we could employ input states entangled with ancilla systems to test them. 
These two measures can be different~\cite{kitaev_quantum_1997}, even for entanglement-breaking channels~\cite{sacchi_entanglement_2005}. 
In Appendix~\ref{app:entcb} we give an example which shows that this is also true of quantum measurements, a specific kind of entanglement-breaking channel. 

Entropic quantities are another means of comparing two probability distributions, an approach taken recently by Buscemi \emph{et al.}~\cite{buscemi_noise_2014} and Coles and Furrer~\cite{coles_state-dependent_2015} (see also Martens and de~Muynck~\cite{martens_disturbance_1992}). 
Both contributions formalize error and disturbance in terms of relative or conditional entropies, and derive their results from entropic uncertainty relations for state preparation which incorporate the effects of quantum entanglement~\cite{berta_uncertainty_2010,tomamichel_uncertainty_2011}. 
They differ in the choice of the entropic measure and the choice of the state on which the entropic terms are evaluated. 
Buscemi \emph{et al.}\ find state-independent error-disturbance relations involving the von Neumann entropy, evaluated for input states which describe observable eigenstates chosen uniformly at random. 
As described in Sec.~\ref{sec:app}, the restriction to uniformly-random inputs is significant, and leads to a characterization of the average-case behavior of the device (averaged over the choice of input state), not the worst-case behavior as presented here. 
Meanwhile, Coles and Furrer make use of general R\'enyi-type entropies, hence also capturing the worst-case behavior. 
However, they are after a state-dependent error-disturbance relation which relates the amount of information a measurement device can extract from a state about the results of a \emph{future} measurement of one observable to the amount of disturbance caused to other observable. 

An important distinction between both these results and those presented here is the quantity appearing in the uncertainty bound, i.e.\ the quantification of complementarity of two observables. 
As both the aforementioned results are based on entropic state-preparation uncertainty relations, they each quantify complementarity by the largest overlap of the eigenstates of the two observables. 
This bound is trivial should the two observables share an eigenstate. 
However, a perfect joint measurement is clearly impossible even if the observables share all but two eigenvectors (if they share all but one, they necessarily share all eigenvectors). 
All three complementarity measures used here are nontrivial whenever not all eigenvectors are shared between the observables.  

\section{Conclusions}
\label{sec:openquestions}

%stress again that the point of uncertainty relations is to "undo" the conclusion of the uncertainty relation --- they give us a way to infer one property from a non-commuting property. The reason we can do this is that we use the structure of the theory. 

%the point of having a nice bound is that it tells you something about the compatibility of observables. just computing the convex hull of the feasible set doesn't give you any particular complementarity measure.

We have formulated simple, operational definitions of error and disturbance based on the probability of distinguishing the actual measurement apparatus from the relevant ideal apparatus by any testing procedure whatsoever. 
The resulting quantities are conceptually straightfoward properties of the measurement apparatus, not any particular fixed experimental setup. 
We presented uncertainty relations for both joint measurability and the error-disturbance tradeoff, for both arbitrary finite-dimensional systems and for position and momentum. 
In the former case the bounds involve simple measures of the complementarity of two observables, while the latter involve the ratio of the desired position and momentum precisions $\sigma_Q$ and $\sigma_P$ to Planck's constant $\hbar$. 
We further showed that this operational approach has  applications to quantum information processing and to wave-particle duality relations. Finally, we presented a detailed comparison of the relation of our results to previous work on uncertainty relations. 

Several interesting questions remain open. 
One may inquire about the tightness of the bounds. The qubit example for conjugate observables discussed at the end of \S\ref{sec:definitions} shows that the finite-dimensional bounds of Theorem \ref{thm:finiteprepdist} are tight for small error $\eps_X$, though no conclusion can be drawn from this example for small preparation disturbance. 
It would be interesting to check the tightness of the position and momentum bounds by computing the error and disturbance measures for a device described by a covariant measurement. 
For reasons of simplicity, we have not attempted to incorporate precision limits into the definitions of error and distinguishability of position and momentum. 
Doing so would lead to more conceptually satisfying bounds and perhaps remedy the fact that the measurement error-preparation disturbance bound is nontrivial even outside the Heisenberg limit. 
Bounds for other observables in infinite dimensions would also be quite interesting, for instance the mixed discrete/continuous case of energy and position of a harmonic oscillator.
Restricting to covariant measurements, in finite or infinite dimensions, it would also be interesting to determine if entangled inputs improve the distinguishability measures, or whether calibration testing is sufficient. From the application in Corollary~\ref{cor:leakage}, it would appear that calibration is sufficient, but we have not settled the matter conclusively.

\vspace{2mm}
{\bf Acknowledgements:} 
The authors are grateful to David Sutter, Paul Busch, Omar Fawzi, Fabian Furrer, Michael Walter and especially David Reeb and Reinhard Werner for helpful discussions.  
This work was supported by the Swiss National Science Foundation (through the NCCR `Quantum Science and Technology' and grant no. 200020-135048) and the European Research Council (grant no.\ 258932). SH is funded by the German Excellence Initiative and the European
Union Seventh Framework Programme under grant agreement no.\ 
291763. He acknowledges additional support by DFG project no.\ 
K05430/1-1.

\printbibliography[heading=bibintoc,title=References]

\appendix
\section{Entanglement improves the distinguishability of measurements}
\label{app:entcb}
Here we give an example of two measurements whose distinguishability is improved by entanglement. 
Let $\cE_1$ be a measurement in an arbitrary chosen basis $\ket{b_0}$, $\ket{b_1}$, and $\ket{b_2}$, and define $\cE_2$ be measurement in the basis given by $\ket{\theta_0}=\frac13(2\ket{b_0}+2\ket{b_1}-\ket{b_2})$, $\ket{\theta_1}=\frac13(-1\ket{b_0}+2\ket{b_1}+2\ket{b_2})$ and $\ket{\theta_2}=\frac13(2\ket{b_0}-\ket{b_1}+2\ket{b_2})$. 
Using  $T_k=\ketbra{b_k}-\ketbra{\theta_k}$, the largest distinguishability to be had without entanglement is given by 
\begin{subequations}
\begin{align}
\delta'(\cE_1,\cE_2)
&=\max_\rho \tfrac12\sum_{k=0}^2\big|\tr[\rho T_k]\big|\\
&=\max_\rho \max_{\{s_k=\pm 1\}} \tfrac12\sum_{k=0}^2 \tr[s_k T_k\rho]\\
&=\max_{\{s_k=\pm 1\}} \big\| \sum_{k=0}^2 s_k T_k \big\|_\infty\,.
\end{align}
\end{subequations}
Checking the eight combinations of $s_k$, one easily finds that the maximimum value is $\sqrt 5/3$. 

Meanwhile, if we use the state 
\begin{align}
\rho=\frac 16\begin{pmatrix} 2 & -1 & -1\\ -1 & 2 & -1\\ -1 & -1 & 2\end{pmatrix}
\end{align}
to define ${\Psi}=(\id\otimes \sqrt{\rho})\Omega(\id\otimes \sqrt{\rho})$ for $\Omega$ the projector onto $\ket{\Omega}=\sum_k \ket{b_k}\otimes \ket{b_k}$, then 
\begin{align}
\delta(\cE_1,\cE_2)\geq \tfrac 12\sum_{k=0}^2\big\|\tr_1[(T_k \otimes \id)\Psi]\big\|_1.
\end{align}
Direct calculation shows that $\delta(\cE_1,\cE_2)\geq\sqrt3/2$.
Thus, there exist projective measurements for which $\delta(\cE_1,\cE_2)> \delta'(\cE_1,\cE_2)$.

%Question: what about conjugate measurements? Answer: same conclusion holds; there is a gap.

\section{Computing error and disturbance by convex optimization}
\label{app:sdp}
Here we detail how to compute the error and disturbance quantities via semidefinite programming and calculate these for the nonideal qubit $X$ measurement example. 
Given a Hilbert space $\cH$ with basis $\{\ket{k}\}_{k=1}^d$, define, just as above, $\ket\Omega=\sum_{k=1}^d\ket{k}\otimes \ket{k}\in \cH\otimes\cH$. 
Then, for any channel  $\cE$, %:\cS(\cH_A)\to \cS(\cH_B)$, 
let $\sC$ denote the Choi mapping of $\cE^*$ to an unnormalized bipartite state,
\begin{align}
  \sC(\cE):=\cE^*\otimes\cI(\ketbra\Omega)\in \cB(\cH_B\otimes \cH_A)\,.
\end{align}
The action of the channel can be compactly expressed in terms of the Choi operator as $\cE_{A\to B}(\Lambda_B)=\tr_B[\Lambda_B \sC(\cE)_{BA}]$ or in the Schr\"odinger picture as $\cE_{A\to B}^*(\rho_A)=\tr_A[\sC(\cE)_{BA}\rho_A^T]$, where the transpose is taken in the basis defining $\sC$ (see, e.g.\ \cite{wolf_quantum_2012}).  
The cb norm can then be expressed in primal and dual form as~\cite{watrous_semidefinite_2009}

\begin{align}
\label{eq:sdp}
\tfrac12\cbnorm{\cE_{A\to B}} 
&=
\begin{aligned}[t]
& \maximum_{K,\rho}
& & \tr[\sC(\cE)_{BA}K_{BA}]\\
& \text{subject to}
& & K_{BA}-\id_B\otimes \rho_A\leq 0, \,\,\tr[\rho_A]\leq 1,\\
&&& \rho_A,K_{BA} \geq 0,
\end{aligned}\\[2mm]
&=
\begin{aligned}[t]
& \minimum_{T,\lambda}
& & \lambda\\
& \text{subject to}
& & T_{BA}\geq \sC(\cE)_{BA},\,\,\lambda \id_A-T_{A}\geq 0,\\
&&& T_{BA},\lambda \geq 0\,.
\end{aligned}
\label{eq:dual}
\end{align}
Note that in the dual formulation the objective function is just the operator norm $\opnorm{T_A}$. 
For infinite-dimensional systems the Choi operator does not have such a nice form, though it might be possible to formulate the cb norm of Gaussian channels as a tractable optimization.

The additional optimizations involving $\cR$ in the measures of error and disturbance are immediately compatible with the dual formulation in \eqref{eq:dual}, and so these quantities can be cast as semidefinite programs. 
To start, consider the error in measuring $X$. 
With $Q_{\sX A}=\sC(\cQ_X)$ and $E_{\sY BA}=\sC(\cE_{A\to \sY B})$, we have
\begin{align}
\label{eq:epsmin}
\eps_X(\cE_{A\to \sY B}) 
&=
\begin{aligned}[t]
& \minimum_{T,\lambda,R}
& & \lambda\\
& \text{subject to}
& & T_{\sX A}+\tr_{\sY}[R_{\sX\sY} E_{\sY  A}]\geq Q_{\sX A},\,\,\lambda\id_A -T_{A}\geq 0, \,\,R_{\sY}= \id_\sY,\\
&&& \lambda,T_{\sX A}, R_{\sX\sY} \geq 0\,.
\end{aligned}
\end{align}
Without loss of generality, we may restrict the  operator $T_{\sX A}$ to be  a hybrid classical-quantum operator, classical on $\sX$, and of course $R_{\sX\sY}$ is classical on both systems. 
This is also the reason it is unnecessary to transpose $\sY$ in $\tr_{\sY}[R_{\sX\sY}E_{\sY A}]$. 
Further symmetries of $Q_{\sX A}$ and $E_{\sX A}$ can be quite helpful in simplifying the program, but we will not pursue this further here. 
The associated primal form is as follows.
\begin{align}
\label{eq:epsmax}
\eps_X(\cE_{A\to \sY B}) 
&=
\begin{aligned}[t]
& \maximum_{K,\rho,L}
& & \tr[Q_{\sX A}K_{\sX A}]-\tr[L_\sY]\\
& \text{subject to}
& & K_{\sX A}-\id_\sX\otimes \rho_A\leq 0,\,\, \tr[\rho_A]\leq 1,\,\,\tr_{A}[E_{\sY  A}K_{\sX A}]- L_{\sY}\otimes \id_{\sX}\leq 0,\\
&&& \rho_A,K_{\sX A}\geq 0, L_{\sY}=L_{\sY}^*.
\end{aligned}
\end{align}
In writing an equality we have assumed that the duality gap is zero. 
But this is easy enough to show using the Slater condition, namely by ensuring that the value of the minimization is finite and that there exists a strictly feasible set of maximization variables. 
The former holds because $\eps_X$ is the infinimum of the distinguishability, and hence $\eps_X(\cE)\geq 0$. 
Meanwhile, a strictly feasible set of variables in \eqref{eq:epsmax} is given by $K=\tfrac12 k\id$, $\rho=k\id$, and $L=k E_{\sY}$ for $k<1/\dim(A)$. 

To formulate the measurement disturbance $\nu_Z(\cE_{A\to \sY B})$ we are interested in $\sC(\cE\circ \cR\circ \cT_\sY\circ \cQ_Z)$, which can be expressed as a linear map on $R_{AB\sY}$:
\begin{subequations}
\begin{align}
\sC(\cE\circ \cR\circ \cT_{\sY}\circ\cQ_Z)&=
\tr_{A'\sY B}[Q_{\sZ A'}R_{A'\sY B }^{T_{A'}}E_{\sY B A}^{T_{B}}]\\
&=\tr_{A'\sY B}[R_{A'\sY B}Q_{\sZ A'}^{T_{A'}}E_{\sY B A}^{T_{B}}]\,.
\end{align} 
\end{subequations}
In the second step we have transposed the $A'$ system in the first. 
Then we have
\begin{align}
\label{eq:numin}
\nu_Z(\cE_{A\to \sY B}) 
&=
\begin{aligned}[t]
& \minimum_{T,\lambda,R}
& & \lambda\\
& \text{subject to}
& & T_{\sZ A}+\tr_{A'\sY B}[R_{A'\sY B} Q_{\sZ A'}^{T_{A'}}E_{\sY B A}^{T_{B}}]\geq Q_{\sZ A},\,\,\lambda\id_A -T_{A}\geq 0, \,\,R_{\sY B}=\id_{\sY B},\\
&&& \lambda,T_{\sZ A}, R_{A'\sY B} \geq 0\,,
\end{aligned}\\[2mm]
\label{eq:numax}
%\nu_Z(\cE_{\sY B|A}) 
&=
\begin{aligned}[t]
& \maximum_{K,\rho,L}
& & \tr[Q_{\sZ A}K_{\sZ A}]-\tr[L_{\sY B}]\\
& \text{subject to}
& & K_{\sZ A}- \id_Z{\otimes} \rho_A\leq 0,\,\,\tr[\rho_A]\leq 1,\,\, \tr_{\sZ A}[Q_{\sZ A'}E_{\sY B A}K_{\sZ A}]- \id_{A'}{\otimes} L_{\sY B}\leq 0,\\
&&& \rho_A,K_{\sZ A}\geq 0,L_{\sY B}=L_{\sY B}^*\,.
\end{aligned}
\raisetag{58pt}
\end{align}
Here we have absorbed the transposes over $A'$ and $B$ into $\id_{A'}$ and the definition of $L_{\sY B}$, since this does not affect Hermiticity or the value of the objective function. 
Strong duality is essentially the same as before: The minimization is finite and we can choose $K=\tfrac12 k\id$ and $\rho=k\id$. 
Then in the third constraint we have $\tr_{\sZ A}[Q_{\sZ A'}E_{\sY B A}K_{\sZ A}]=\tfrac12 k\id_{A'}\otimes E_{\sY B}$ since $\cQ_Z$ is unital. 
Setting $L=k E_{\sY B}$ gives a strictly feasible set. 

%Defining $P_{A\sZ}=\sC(\cP_{Z})$, $\sC(\cC)=\sigma_{B\sX}\otimes \id_\sZ$, and $R_{AB\sX}$ for the Choi representation of the recovery map $\cR$, the Choi representative of $\cP_Z\circ \cE$ is simply $\sC(\cP_Z\circ \cE)=\tr_{A}[E_{B\sX A}P_{A\sZ}^{T_A}]$.
%Denoting this operator by $D_{B\sX \sZ}$, we obtain the following expressions for the two preparation disturbance measures. 
Finally, we come to the two preparation disturbance measures. 
The first is simply 
\begin{align}
\label{eq:etamin}
\eta_Z(\cE_{A\to \sY B}) 
&=
\begin{aligned}[t]
& \minimum_{T,\lambda,R}
& & \lambda\\
& \text{subject to}
& & T_{A \sZ}+\tr_{\sY BA'}[R_{A\sY B}  E_{\sY BA'}^{T_{B}} P_{A'\sZ}^{T_{A'}} ]\geq P_{A\sZ},\,\,\lambda \id_\sZ-T_{\sZ}\geq 0, \,\,R_{\sY B}=\id_{\sY B},\\
&&& \lambda,T_{\sX A}, R_{A\sY B} \geq 0\,,
\end{aligned}\\[2mm]
\label{eq:etamax}
%\eta_Z(\cE_{\sY B|A}) 
&=
\begin{aligned}[t]
& \maximum_{K,\rho,L}
& & \tr[P_{A\sZ}K_{A \sZ}]-\tr[L_{\sY B}]\\
& \text{subject to}
& & K_{A\sZ}-\id_A\otimes \rho_\sZ\leq 0,\,\,\tr[\rho_\sZ]\leq 1,\,\, \tr_{A'\sZ}[E_{\sY BA'}P_{A'\sZ}^{T_{A'}}K_{A\sZ}]-\id_A\otimes L_{\sY B}\leq 0,\\
&&& \rho_\sZ,K_{A\sZ}\geq 0,L_{\sY B} =L_{\sY B}^*\,.
\end{aligned}
\raisetag{58pt}
\end{align}
Here we have absorbed the transpose on $B$ into the definition of $L_{\sY B}$ since this doesn't affect Hermiticity or the value of the objective function. 
Strong duality holds as before, and also for the demerit measure which reads %, this time relying on the fact that $\tr_Z [P_{A\sZ}]=\id_A$ for state preparation which prepares states from an orthonormal basis. 
\begin{align}
\label{eq:wideetamin}
\tfrac{d-1}d-\widehat\eta_Z(\cE_{A\to \sY B}) 
&=
\begin{aligned}[t]
& \minimum_{T,\lambda,\sigma}
& & \lambda\\
& \text{subject to}
& & T_{\sY B\sZ}+\sigma_{\sY B}{\otimes} \id_\sZ\geq \tr_A [E_{\sY BA}P_{A\sZ}^{T_A}],\,\,\lambda \id_\sZ-T_{\sZ}\geq 0, \,\,\tr[\sigma_{\sY B}]=1,\\
&&& \lambda,T_{\sY B\sZ}, \sigma_{\sY B} \geq 0\,,
\end{aligned}
\raisetag{58pt}\\[2mm]
\label{eq:wideetamax}
%\tfrac{d-1}d-\widehat\eta_Z(\cE_{\sY B|A}) 
&=
\begin{aligned}[t]
& \maximum_{K,\rho,\mu}
& & \tr[E_{\sY BA}P_{A\sZ}^{T_A}K_{\sY B \sZ}]-\mu\\
& \text{subject to}
& & K_{\sY B\sZ}-\id_{\sY B}\otimes \rho_\sZ\leq 0,\,\,\tr[\rho_\sZ]\leq 1,\,\, K_{\sY B}-\mu\id_{\sY B}\leq 0,\\
&&& \rho_\sZ,K_{\sY B\sZ}\geq 0,\mu\in \mathbb R\,.
\end{aligned}
\end{align}

Now let us consider the particular example described in the main text, a suboptimal $X$ measurement. 
Suppose we use $\ket{\phi_x}$ from the ideal $X$ measurement to define the Choi operator. 
After a bit of calculation, one finds that the Choi operator $E_{\sY BA}$ of $\cE_{\sY B|A}$ is given by 
\begin{align}
E_{\sY B A}=\ketbra{b_0}_\sY\otimes \ketbra{\Psi}_{BA}+\ketbra{b_1}_\sY\otimes (\sigma_z\otimes \sigma_z)\ketbra{\Psi}_{BA}(\sigma_z\otimes \sigma_z),
\end{align}
where $\ket{\Psi}=\cos\tfrac\theta 2\ket{\phi_0}\otimes\ket{\phi_0}+\sin\tfrac\theta 2\ket{\phi_1}\otimes\ket{\phi_1}$.
Tracing out $B$ gives the Choi operator of just the measurement result $\sY$, $E_{\sY A}=\sum_x \ketbra{b_x}_\sY\otimes \Lambda_x$, with $\Lambda_x=\tfrac12\mathbbm{1}+\tfrac12(-1)^x\cos \theta\,\, \sigma_x$. 

To compute the measurement error $\eps_X(\cE)$, suppose that no recovery operation is applied, i.e.\ the outcome $\sY$ is treated as $\sX$. 
Then we can work with $E_{\sX A}$ and dispense with $R$ so that the third constraint in \eqref{eq:epsmin} is satisfied.
To satisfy the first constraint, choose $T_{\sX A}$ to be the positive part of $Q_{\sX A}-E_{\sX A}$.
This gives $T_{\sX A}=\tfrac12(1-\cos\theta)\sum_x \ketbra{b_x}\otimes \ketbra{\varphi_x}$; consequently, $T_{A}=\tfrac12(1-\cos\theta)\id_A$ and therefore $\eps_X(\cE)\leq \tfrac12(1-\cos\theta)$. 
On the other hand, $K_{\sX A}=\tfrac12Q_{\sX A}$ and $\rho_A=\tfrac 12\id_A$ satsify the first two constraints in \eqref{eq:epsmax}. 
The last constraint involves the quantity $\tr_A[E_{\sY A}K_{\sX A}]=\tfrac14\sum_{xy}\ketbra{b_y}_{\sY}\otimes \ketbra{b_x}_{\sX}(1+(-1)^{x+y}\cos\theta)$ and can therefore be satisfied by choosing $L_{\sY}=\tfrac14(1+\cos\theta)\id_{\sY}$. 
Evaluating the objective function gives $\eps_X(\cE)\geq \tfrac12(1-\cos\theta)$. 

Note that the choice of $K_{\sX A}$ corresponds to the unentangled test of randomly inputting $\ket{\varphi_x}$ and checking that the result is $x$. 
We could have anticipated that unentangled tests would be sufficient in this case, since the optimal and actual measurements are both diagonal in the $\sigma_x$ basis: Any input state can be freely dephased in this basis, thus removing any entanglement.

Next, consider the measurement disturbance $\nu_Z(\cE)$. 
Proceeding as with measurement error, suppose that no recovery operation is applied, so that the output $B$ is just regarded as $A'\simeq A$ and the third constraint in \eqref{eq:numin} is trivially satisfied.
For the first constraint we need only the operator 
$\tr_{\sY A'}[Q_{\sZ A'}E_{\sY A' A}^{T_{A'}}]$, and after some calculation we find that it equals $\sum_z \ketbra {b_z}\otimes \Gamma_z$ with $\Gamma_z=\tfrac12(\id+(-1)^z\sin\theta \sigma_z)$. 
Thus, the optimization is just like that of $\eps_X(\cE)$, but with $\cos\theta$ replaced by $\sin\theta$. 
Hence $\nu_Z(\cE)\leq \tfrac 12(1-\sin\theta)$. 
To show the other inequality from the maximization form \eqref{eq:numax} also proceeds as before, starting with $K_{\sZ A}=\tfrac12Q_{\sZ A}$ and $\rho_A=\tfrac 12\id_A$. 
For the third constraint a bit of calculation shows 
\begin{align}
\label{eq:Fcontract}
\tr_{\sZ A}[Q_{\sZ A'}E_{\sY B A}K_{\sZ A}]
&=\tfrac14\sum_{x,z} \ketbra{b_z}_{A'}\otimes \ketbra{b_x}_\sY\otimes (\sigma_x^z\sigma_z^x\ketbra\psi\sigma_z^x\sigma_x^z)_B,
\end{align}
with $\ket{\psi}=\tfrac1{\sqrt2}(\sqrt{1+\sin\theta}\ket{\theta_0}+\sqrt{1-\sin\theta}\ket{\theta_1}$.
Choosing 
\begin{align}
\label{eq:LBX}
L_{\sY B}=\tfrac18\sum_x \ketbra{b_x}_\sY\otimes ((1+\sin\theta)\id+(-1)^x\cos\theta\,\sigma_x)_B
\end{align} 
satisfies the constraints, and the objective function becomes $\tfrac 12(1-\sin\theta)$. 
As with $\eps_X(\cE)$, entangled inputs do not increase the distinguishability in this particular case. 

A trivial recovery map also optimizes $\eta_Z(\cE)$. 
To see this, set $K_{A\sZ}=\tfrac12 P_{A\sZ}$ and $\rho_\sZ=\tfrac12 \id_{\sZ}$. 
Then in the third constraint of \eqref{eq:etamax} we have $\tr_{A'}[E_{\sY BA'} P_{A'\sZ}^{T_{A'}}K_{A\sZ}]$, which is precisely the same as \eqref{eq:Fcontract} with $A'$ replaced by $A$. 
Hence, if we choose $L_{\sY B}$ as in \eqref{eq:LBX}, we obtain the lower bound $\eta_Z(\eps)\geq \tfrac12(1-\sin\theta)$. 
To establish optimality, suppose $\cR$ does nothing but discard the $\sY$ system. 
In the minimization \eqref{eq:etamin} we then have $\tr_{\sY A'}[E_{\sY AA'}P^{T_{A'}}_{A'\sZ}]$, which is the same as $\tr_{\sY A'}[Q_{\sZ A'}E_{\sY A' A}^{T_{A'}}]$ from $\nu_Z(\cE)$. 
Proceeding as there, we find the matching upper bound. 

Finally, consider $\widehat{\eta}_Z(\cE)$. 
Here there are two possible outputs of $\cP_Z\circ\cE$, call them $\xi_0$ and $\xi_1$. 
It is not difficult to show that for arbitrary $\xi_z$ the distinguishability is precisely $\widehat \eta_Z(\cE)=\tfrac12(1-\delta(\xi_0,\xi_1))$.
On the one hand, we can simply pick the output of $\cC$ to be $\xi=\tfrac12(\xi_0+\xi_1)$. 
Then, with $T$ in \eqref{eq:wideetamin} the positive part of $\sum_z \ketbra{z}_{\sZ}\otimes (\xi_z-\xi)$, the objective function becomes $\tfrac12\delta(\xi_0,\xi_1)$. 
On the other hand, in \eqref{eq:wideetamax} we can choose $K_{\sY B\sZ}=\tfrac12\ketbra{0}_\sZ\otimes \Lambda_{\sY B}+\tfrac12\ketbra{1}_\sZ\otimes (\id-\Lambda)_{\sY B}$, for $\Lambda$ the projector onto the nonnegative part of $\xi_0-\xi_1$. 
Then $\mu=\tfrac12$ and $\rho=\tfrac12\id$ are feasible and lead again to the same objective function. 
In this particular case the two states are $\xi_1=\sigma_z\xi_0\sigma_z$ and $\xi_0=\tfrac12\sum_x \ketbra{b_x}_\sY\otimes \sigma_x \ketbra\psi\sigma_x$, which yields $\delta(\xi_0,\xi_1)=\sin\theta$ and hence $\widehat \eta_Z(\cE)=\tfrac12(1-\sin\theta)$.

\section{Counterexample channel}
\label{app:counterexample}
Here we present the calculations involved in \S\ref{sec:appinfodisturb}
Let $\ket{\xi_z}_{BD}=\sum_y \sqrt{p_{yz}}\ket{y}_B\ket y_D$. Then the isometry is just
\begin{align}
V=\sum_z \ket{\xi_z}\ket z_C\bra z_A.
\end{align}
Observe that the action on $\ket{\tilde x}$ states leads to symmetric output in $BC$:
\begin{subequations}
\begin{align}
V\ket{\tilde x}
&=\tfrac1{\sqrt d}\sum_z \omega^{xz}V\ket{z}\\
&=\tfrac1{\sqrt d}\sum_z \omega^{xz}\ket{\xi_z}_{BD}\ket z_C\\
&=Z^x_C\, \tfrac1{\sqrt d}\sum_z \ket{\xi_z}_{BD}\ket z_C.
\end{align}
\end{subequations}
Therefore, the probability of incorrectly identifying any particular input state is the same as any other, and we can consider the case that the input $x$ value is chosen uniformly at random.
We can further simplify the $BC$ output by defining $p_y=\tfrac1d \sum_z p_{yz}$ and 
\begin{align}
\ket{\eta_y}=\tfrac1{\sqrt d}\sum_z \sqrt{p_{yz}/p_y}\ket{z},
\end{align}
which is a normalized state on $\cH_C$ for each $y$. 
Then we have 
\begin{align}
V\ket{\tilde x}=Z^x_C\sum_y \sqrt{p_y}\ket{\eta_y}_C\ket{y}_B\ket y_D.
\end{align}
Ignoring the $D$ system will produce a classical-quantum state, with system $B$ recording the classical value $y$, which occurs with probability $p_y$, and $C$ the quantum state $Z^x\ket{\eta_y}$.
The optimal measurement therefore has elements $\Lambda_x$ of the form $\Lambda_x=\sum_y \ketbra y_B\otimes (\Gamma_{x,y})_C$ for some set of POVMs $\{\Gamma_{x,y}\}_y$. 
In every sector of fixed $y$ value, the measurement has to distinguish between a set of pure states occurring with equal probabilities. 
Therefore, by a result going back to Belavkin, the optimal measurement is the so-called ``pretty good measurement''~\cite{belavkin_optimal_1975,hausladen_`pretty_1994}. 
This has measurement elements $\Gamma_{x,y}$ which project onto the orthonormal states $\ket{\mu_{x,y}}=S^{-1/2}Z^x\ket{\eta_y}$, where $S=\sum_x Z^x\ketbra{\eta_y}Z^{-x}$. 
It is easy to work out that $S=\sum_x (p_{yz}/p_y)\ketbra z$, and thus $\ket{\mu_{x,y}}=\ket{\tilde x}$ for all $y$. 
Hence, we can in fact dispense with the $B$ system altogether, since the particular value of $y$ does not alter the optimal measurement. 
The average guessing probability is thus 
\begin{subequations}
\begin{align}
p_{\rm guess}
&=\tfrac 1d \sum_{x,y}p_y \big|\bra{\tilde x}Z^x\ket{\eta_y}\big|^2\\
&=\sum_y p_y\big|\braket{\tilde 0|\eta_y}\big|^2\\
&=\tfrac 1{d^2}\sum_y \Big(\sum_z \sqrt{p_{yz}}\Big)^2,
\end{align}
\end{subequations}
as intended.

\section{Englert's complementarity relation}
\label{app:englert}
Here we describe Englert's setup in our formalism and establish \eqref{eq:comprelations}. 
He considers a Mach-Zehnder interferometer with a relative phase shift between the two arms and additional which-way detectors in each arm. 
To the two possible paths inside the interferometer we may associate the (orthogonal) eigenstates $\ket{\vartheta_z}$ of an observable $Z$, with $z\in\{0,1\}$. 
For simplicity, we assume $Z$ has eigenvalues $(-1)^z$. The action of a relative $\phi$ phase shift is described by the unitary $U_{\rm PS}=\sum_{z=0}^1e^{iz\phi}\ketbra{\vartheta_z}$. 
It will prove convenient to choose $\phi=0$ below, but we leave it arbitrary for now. 
Meanwhile, the which-way detectors can be described as producing different states of an ancilla system, depending on which path the photon takes. 
For pure ancilla states $\ket{\gamma_z}$, the detector corresponds to the isometry $U_{\rm WW}=\sum_{z=0}^1\ketbra{\vartheta_z}_Q\otimes \ket{\gamma_z}_A$, where $A$ denotes the ancilla and $Q$ the system itself, which Englert terms a ``quanton''. 

Ignoring the phase shifts associated with reflection, the output modes of a symmetric (50/50) beamsplitter are related to the input modes by the unitary $U_{\rm BS}=\sum_{z=0}^1 \ket{\vartheta_z}\bra{\varphi_z}$, with $\ket{\varphi_x}=\tfrac1{\sqrt{2}}\sum_z(-1)^{xz}\ket{\vartheta_z}$ for $x\in\{0,1\}$. 
We may associate these states with the observable $X$, also taking eigenvalues $(-1)^x$. 
Observe that all three complementarity measures are $\tfrac12$. 
The entire Mach-Zehnder device can be described by the isometry
\begin{subequations}
\begin{align}
U_{\rm MZ}&=U_{\rm BS}U_{\rm PS}U_{\rm WW}U_{\rm BS}\\
&=\sum_{x,z=0}^1e^{ix\phi}\ket{\vartheta_z}\braket{\varphi_z|\vartheta_x}\bra{\varphi_x}_Q\otimes \ket{\gamma_x}_A\\
&=\sum_{x=0}^1e^{ix \phi}\ketbra{\varphi_x}_Q\otimes \ket{\gamma_x}_A.
\end{align}
\label{eq:umz}
\end{subequations}
When the ancilla is subsequently measured so as to extract information about the path, we may regard the whole operation as an apparatus $\cE_{\rm MZ}$ with one quantum and one classical output.

The available ``which-way'' information, associated with particle-like behavior of $Q$, is characterized by the distinguishability $\cD:=\delta(\gamma_0,\gamma_1)$. 
Given the particular form of $U$ in \eqref{eq:umz}, we may set $\sin\theta=\braket{\gamma_0|\gamma_1}$ for $\theta\in\mathbb{R}$ without loss of generality; $D$ is then $\cos\theta$. 
This amounts to defining $\ket{\gamma_k}=\cos\frac\theta2\ket{k}+\sin\frac\theta2\ket{k+1}$, where the states $\{\ket{k}\}_{k=0}^1$ form an orthonormal basis and arithmetic inside the ket is modulo two. 
Thus, $\cE_{\text{MZ}}$ with $\phi=0$ is precisely the nonideal qubit $X$ measurement $\cE$ considered in \S\ref{sec:definitions}. 
We shall see momentarily that $\phi=0$ can be chosen without loss of generality. 
Using \eqref{eq:MZerror} we have  $\eps_X(\cE_{\rm MZ})=\tfrac12(1-\cD)$ as claimed.

Meanwhile, the fringe visibility $\cV$ is defined as the difference in probability (or population) in the two output modes of the interferometer, maximized over the choice of input state. 
Since $Z=\ketbra{\vartheta_0}-\ketbra{\vartheta_1}$, this is just
\begin{align}
\cV=\max_\rho \big|\tr[(Z_Q\otimes \id_A)U_{\rm MZ}\rho U_{\rm MZ}^*]\big|\,.
\end{align}
A straightforward calculation yields $U_{\rm MZ}^* (Z_Q\otimes \id_A)U_{\rm MZ}=\sin\theta (\cos\phi\, Z+i\sin\phi\,XZ)$. 
It can be verified that $(\cos\phi\, Z+i\sin\phi\,XZ)$ has eigenvalues $\pm 1$, and therefore $\cV=\sin\theta$. Thus, $\cV^2+\cD^2=1$ in this case (cf.\ \cite[Eq.~11]{englert_fringe_1996}). Note that $\phi$ does not appear in the visiblity itself, justifying our choice of $\phi=0$ above. 
By \eqref{eq:MZdisturbance},  $\nu_Z(\cE_{\rm MZ})=\eta_Z(\cE_{\rm MZ})=\widehat\eta_Z(\cE_{\rm MZ})=\tfrac12(1-V)$. 

\end{document}